\documentclass[onecolumn,secnumarabic,amssymb, nobibnotes, aps, prd]{revtex4-2}

\usepackage{amsmath,amsfonts,amsthm,graphicx,gensymb,float,algorithm,braket,xcolor,amssymb,xparse,array,algpseudocode}

\usepackage[caption=false]{subfig}

\usepackage{dblfloatfix}

\newtheorem{lemma}{Lemma}

\begin{document}

\title{A Proposed Quantum Hamiltonian Encoding Framework for Time Evolution Operator Design of Potential Energy Function}%

\author{Mostafizur Rahaman Laskar$^{1}$,~ Kalyan Dasgupta$^2$, Atanu Bhattacharya $^3$}
\affiliation{$^1$G. S. Sanyal School of Telecommunications, Indian Institute of Technology, Kharagpur, \\$^2$IBM Research, Bangalore, India,\\
$^3$Department of Chemistry, GITAM, Visakhapatnam.}

\begin{abstract}

The exploration of potential energy operators in quantum systems holds paramount significance, offering profound insights into atomic behaviour, defining interactions, and enabling precise prediction of molecular dynamics. By embracing the Born-Oppenheimer picture, we delve into the intricate quantum evolution due to potential energy, facilitating accurate modelling and simulation of atomic phenomena with improved quantum fidelity. This research delves into time evolution operation due to potential energy functions for applications spanning quantum chemistry and condensed matter physics. Challenges in practical implementation, encompassing the formidable curse of dimensionality and intricate entangled interactions, are thoughtfully examined. Drawing upon seminal works, we lay a robust foundation for comprehensive investigations into potential energy landscapes with two proposed algorithms. In one methodology, we have shown a systematic decomposition of the potential energy function into Hadamard bases with composite construction of Pauli-$Z$, identity and $R_Z$ gates which can construct the unitary time evolution operator corresponding to the potential energy with a very high fidelity. The other method is a trade-off between complexity and fidelity, where we propose a novel quantum framework that can reduce the gate complexity from $\Theta(2^n)$ to $\Theta(^nC_r)$ (for some $r<n$). The proposed quantum algorithms are capable of efficiently simulating potential energy operators. The algorithms were implemented in simulators and IBM quantum hardware to prove their efficacy.

\end{abstract}
 
\maketitle

\section{Introduction}
Quantum simulation in atomic chemistry holds tremendous importance as it allows us to explore the complex behaviours of atomic particles with unprecedented precision and accuracy. By leveraging potential energy operators in the quantum Hamiltonian, researchers can gain valuable insights into molecular structures, chemical reactions, and material properties \cite{feynman2000theory, aspuru2005simulated, kassal2011simulating, mcclean2016theory}. Quantum simulations provide a powerful tool for studying atomic behaviour, offering applications in quantum chemistry, condensed matter physics, and quantum computation \cite{manzhos2015neural}. The ability to precisely model the quantum evolution of potential energy is essential for advancing our understanding of atomic systems and unlocking transformative applications in diverse scientific disciplines \cite{ollitrault2021molecular, shao2016communication}.

In atomic chemistry, the study of potential energy operators has a rich history and plays a crucial role in the development of quantum mechanical treatments of molecules and solids. Over the years, researchers have employed diverse theoretical and computational methods to explore potential energy landscapes and understand their influence on atomic behaviour \cite{shao2016communication, kaser2023neural}. State-of-the-art literature has highlighted the potential of machine learning techniques, such as neural network potential energy surfaces, to efficiently approximate complex potential energy landscapes \cite{manzhos2020neural}. Furthermore, quantum algorithms have been developed for grid-based variational time evolution and threshold gate-based quantum simulation, enabling accurate simulations of quantum systems with reduced computational resources \cite{ollitrault2022quantum, sornborger2018toward, kivlichan2017bounding, klymko2022real, Cao-QC_QC}. The classical approach to solving the time evolution problem has been to consider a linear response for small time periods. Other approaches have also used tensor networks. The advantage that Quantum computing gives is in the way the matrix exponentials are easily realized by decomposing them into strings of Pauli matrices. For non-commuting Hamiltonian terms, the Trotter-Suzuki method may be applied for the individual terms \cite{Cao-QC_QC}. Present-day Quantum computers face limitations as the number of trotter terms increases. These methods are mainly envisaged for fault-tolerant quantum systems.

The limitations that noisy present-day quantum systems challenge us with make it all the more important to develop quantum algorithms for Hamiltonian (potential and kinetic energy) operators with reduced computational complexity \cite{shokri2021implementation}. As the system size increases, the computational resources required for accurate simulations grow exponentially, leading to the "curse of dimensionality" that hampers exact simulations for large systems \cite{aspuru2005simulated, kassal2011simulating}. To overcome these limitations, novel quantum Hamiltonian encoding algorithms are needed. Efficient encoding schemes that reduce the gate complexity while preserving the accuracy of the simulation will be instrumental in making potential energy simulations accessible on near-term quantum hardware \cite{motta2021low, peruzzo2014variational, ollitrault2021molecular}. By mitigating the computational challenges, these advancements will pave the way for furthering our understanding of atomic behaviour and unlocking the full potential of quantum technologies in various scientific applications \cite{fu2018ab, ollitrault2022quantum}.

Recent progress in quantum computing and machine learning has sparked interdisciplinary research in the development of neural network potential energy surfaces and their integration into quantum simulations \cite{kaser2023neural}. These data-driven models offer efficient and reliable representations of potential energy landscapes, presenting exciting opportunities for accurate simulations of small molecules and chemical reactions \cite{manzhos2020neural}. Additionally, the combination of quantum algorithms with neural network potentials has shown promise in the real-time evolution of ultra-compact Hamiltonian eigenstates on quantum hardware, further advancing the capabilities of quantum simulations \cite{klymko2022real}. Embracing the potential of machine learning and quantum computing in the study of potential energy opens new frontiers for understanding complex atomic systems and designing innovative materials. Preparing a diagonal matrix using the Walsh series approximation of a given function without ancillary qubit is demonstrated in \cite{welch2014efficient}. Extending this idea to quantum chemistry problems, we demonstrate conceptually new quantum methods with a focus towards low-complexity algorithm design with improved fidelity response. 

In this paper, we propose two different methods to encode the potential energy operator in a 1D lattice for the time evolution problem in a quantum circuit. The first method looks at encoding any arbitrary potential surface using diagonal matrix representations of the Hadamard or Walsh-Hadamard basis functions. This is an exact method and requires $2^n$ basis operations. This kind of approach may be more suitable in the fault-tolerant regime. The second method is a parameterized approximate approach, whereby, assumptions are made on the order of the potential surface (linear, quadratic, or some polynomial order). Based on the acceptability of error bounds, one can fit the potential curve by increasing/decreasing the number of phase gates and their order of entanglement. The higher the number of phase gates and more the entanglements, the higher will be the precision of the fit. A perfect fit would of course need $2^n$ gate operations, involving phase gates and entanglers (controlled phase gates). We believe this kind of a formulation to be more apt for noisy quantum machines.

Our paper is organized as follows. In section \ref{sec:PE operator} we give a brief overview of the structure of the Hamiltonian. Since this paper is about encoding the potential energy operator, we focus mainly on the potential landscape in a 1D lattice. In section \ref{sec:proposed_operator} we detail our proposed methods and discuss the encoding methodologies. Section \ref{sec:gate_complexity} discusses the gate complexity of the two proposed methods. In section \ref{sec:results} we give the results and show the ability of the methods to reconstruct the potential surface. We also give some fidelity results from actual hardware runs on IBM quantum machine to see the real-time performance of the proposed framework. Finally, we conclude our article based on our theoretical formulation and practical realization with experimental results.


\section{Potential Energy Operator and Classical Encoding} \label{sec:PE operator}

In the realm of non-relativistic physics, the Hamiltonian governs the intricate behaviour of particles, such as electrons, as they interact with the external potential of other particles, notably positively charged nuclei. This intricate interplay is elegantly described within the Born-Oppenheimer approximation, as represented by the following profound equation:

\begin{align}
\mathbf{H} = -\sum_{i} \frac{\nabla^2_i}{2} - \sum_{i,j} \frac{q_j}{\vert R_j - r_i \vert} + \sum_{i<j} \frac{1}{r_i - r_j} + \sum_{i<j} \frac{q_i q_j}{R_i - R_j}.
\label{diff-Ham}
\end{align}

Here, $\mathbf{H}_1=-\sum_{i} \frac{\nabla^2_i}{2}$ is the kinetic energy term, $\mathbf{H}_2=- \sum_{i,j}\frac{q_j}{\vert R_j -r_i \vert}$ denotes the potential energy where $q_j$ are charges of the nuclei, $R_j$ and $r_i$ are positions of the nuclei and electrons respectively; the term $\mathbf{H}_3=\sum_{i<j}\frac{1}{r_i-r_j}$ denotes the electron-electron repulsion potential term, and $\mathbf{H}_4=\sum_{i<j}\frac{q_iq_j}{R_i-R_j}$ is a constant term. We use the discretization techniques to make the differential form (\ref{diff-Ham}) implementable on a digital computer \cite{babbush2017low,kivlichan2017bounding}, and the effective Hamiltonian is simplified as
\begin{align}
    \mathbf{H} = \mathbf{K}(\hat{p}) + \mathbf{V}(\hat{x}),
\end{align}
where $\mathbf{K}(\hat{p})$ represents the discretized kinetic energy operator as a function of momentum ($\hat{p}$), and the term $\mathbf{V}(\hat{x})$ denotes the potential energy expressed in space coordinates ($\hat{x}$). Note that, we assume all energy terms (except the kinetic energy part) to be included with the potential energy operator, given by 
\begin{align}
\mathbf{V} = -\sum_{i,j} \frac{q_j}{\vert R_j - r_i \vert} + \sum_{i<j} \frac{1}{r_i - r_j} + \sum_{i<j} \frac{q_i q_j}{R_i - R_j}.
\label{diff-Ham}
\end{align}

The time-dependent Schrodinger's equation (TDSE) is given by equation (\ref{eqn:TDSE}).
\begin{equation}
    i\hbar \frac{\partial}{\partial t}\ket{\Psi (t)} = \hat{H}\ket{\Psi(t)} \label{eqn:TDSE}
\end{equation}
The Hamiltonian $\hat{H}$ consists of the potential energy and the kinetic energy part. The time evolution due to $\hat{H}$ is given by \ref{eqn:time_evol} \cite{Townsend}.
\begin{equation}
    \ket{\Psi (t)} = e^{-i\hat{H}t/\hbar}\ket{\Psi(0)} \label{eqn:time_evol}
\end{equation}
The potential energy operates in the spatial dimensions. The state $\ket{\Psi(t)}$ is in the form of a vector with each element representing the probability amplitude of the wave function at a given point in a 1D/2D lattice. Let us consider a 1D lattice, where every lattice point is represented by a state. If we just consider the potential alone, and write $\hat{H} = \mathbf{V}(\hat{x})$, the state evolves due to the action of  $e^{-i\mathbf{V}(\hat{x})t/\hbar}$ on the state vector $\ket{\Psi(t)}$. If the potential function is well-defined (continuous, square-integrable) for every point in the lattice, it can be represented in the form of a diagonal matrix.  The exponential of such a matrix (as used in (\ref{eqn:time_evol})) is also a diagonal matrix. If we consider all the equations in atomic units and ignore constants like $\hbar$, we get the modified equation (\ref{eqn:time_evol_pot}).
\begin{equation}
    \ket{\Psi (t)} = e^{-i\mathbf{V}(\hat{x})t}\ket{\Psi(0)} \label{eqn:time_evol_pot}
\end{equation}
The matrix $e^{-i\mathbf{V}(\hat{x})t}$ in our analysis will be considered a diagonal matrix.

\begin{figure}[htb!]
    \centering
    \includegraphics[width=0.49\linewidth]{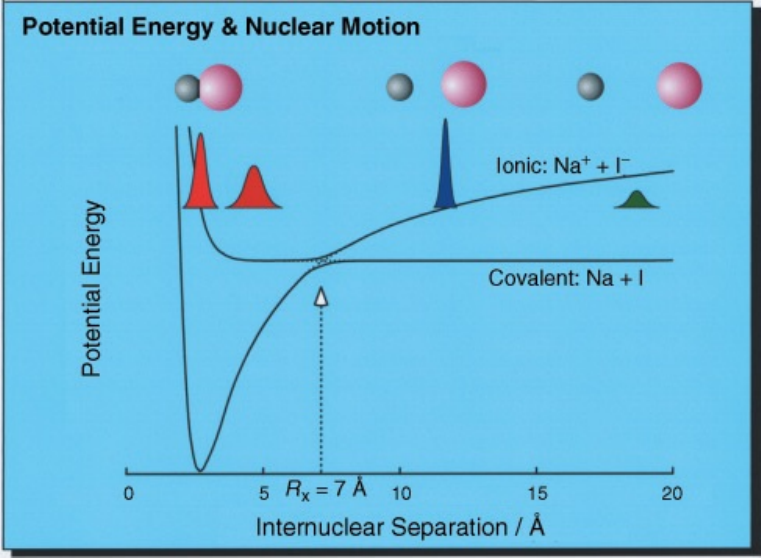}
    \caption{Potential energy curves in Sodium Iodide molecule \cite{baskin2001freezing}}
    \label{NAI_potential}
\end{figure}

In Fig.\ref{NAI_potential}, the potential energy distribution of Sodium Iodide (NaI) as function of atomic distance is shown \cite{baskin2001freezing}.

\begin{figure}[htb!]
\centering 
\includegraphics[width=0.49\linewidth]{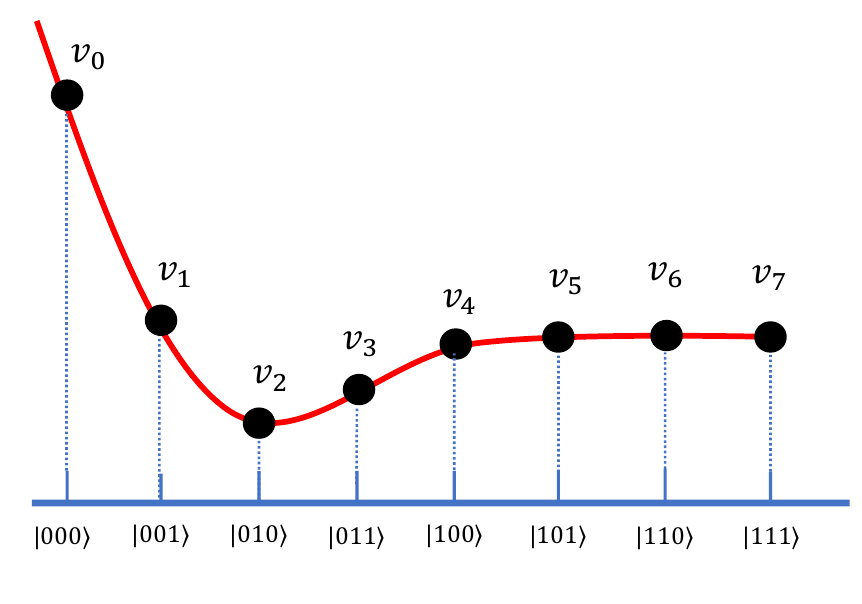} 
\caption{Figure showing the potential curve over the 1D lattice and the corresponding diagonal potential matrix}
\label{fig:potential_states}
\end{figure}

In our quantum circuit, the encoding method we have used is amplitude encoding, where every point in the lattice is a state. The states run from $\ket{0}^{\otimes n}$ to $\ket{1}^{\otimes n}$ for an $n$ qubit system. The leftmost bit corresponds to the most significant qubit (msqb), while the rightmost bit corresponds to the least significant qubit (lsqb). If we consider a 1D lattice, we can think of the layout as given in Fig. \ref{fig:potential_states} for a 3 qubit system. The plot on the left gives the potential curve as experienced over the lattice. The lattice points are marked by the states. The corresponding diagonal matrix is given by
\begin{align}
    \mathbf{V}(\hat{x})=\begin{bmatrix}
        v_0&&&&&&&\\
        &v_1&&&&&&\\
        &&v_2&&&&&\\
        &&&v_3&&&&\\
        &&&&v_4&&&\\
        &&&&&v_5&&\\
        &&&&&&v_6&\\
        &&&&&&&v_7
    \end{bmatrix}
    \label{diagpotential}
\end{align}
where $v_0,v_1,~\dots, v_7$ are the discretized potential energies.

\section{Proposed Quantum Encoding Methods} \label{sec:proposed_operator}

\subsection{Encoding in Hadamard basis}\label{hadamardbasis_theory}

To encode our potential function into basis functions that can be encoded in a quantum circuit, we use the Hadamard transform or the Walsh-Hadamard transform \cite{Mike_Ike}. The matrix generating the Hadamard functions can be obtained by taking the tensor products of Hadamard matrices, $H^{\otimes n}$. For example, if we take the case of 3 qubits, the Hadamard functions are given by the columns (or the rows, as the matrix is symmetric) of the matrix given in (\ref{eqn:Hadamard_func}). As depicted in the equation, the basis functions are $b_1 = \frac{1}{\sqrt{8}}\left[1, 1, 1, 1, 1, 1, 1, 1 \right]^T, b_2 = \frac{1}{\sqrt{8}}\left[1, -1, 1, -1, 1, -1, 1, -1 \right]^T$, and so on.
\begin{equation}
    H^{\otimes 3} = \frac{1}{\sqrt{8}}\left[
\begin{array}{rrrrrrrr}
b_1 & b_2 & b_3 & b_4 & b_5 & b_6 & b_7 & b_8 \\
1 &  1 &  1 &  1 &  1 &  1 &  1 &  1 \\
1 & -1 &  1 & -1 &  1 & -1 &  1 & -1 \\
1 &  1 & -1 & -1 &  1 &  1 & -1 & -1 \\
1 & -1 & -1 &  1 &  1 & -1 & -1 &  1 \\
1 &  1 &  1 &  1 & -1 & -1 & -1 & -1 \\
1 & -1 &  1 & -1 & -1 &  1 & -1 &  1 \\
1 &  1 & -1 & -1 & -1 & -1 &  1 &  1 \\
1 & -1 & -1 &  1 & -1 &  1 &  1 & -1 \\
\end{array} \right]
\label{eqn:Hadamard_func}
\end{equation}
The Hadamard basis functions have rectangular shapes and for approximating any given function defined over an $n$ qubit system, one would need $2^n$ such basis functions. The factor $\frac{1}{\sqrt{2^n}}$ ensures that the bases are orthonormal and the matrix unitary. The entire potential function over the lattice points can be represented as a linear combination of the bases as given in (\ref{lin_comb_had}).
\begin{equation}
    f(x) = \sum_{j=1}^{2^n} c_j b_j(x) \label{lin_comb_had}
\end{equation}
Here $x$ denotes the lattice points. The coefficient $c_j \in \mathbb{R}$ can be obtained by taking the inner product of the function $f(x)$ with the basis $b_j(x)$.
\begin{equation}
    c_j = b_j(x)^T.f(x)
\end{equation}

We show the reconstruction of a function given by $f(x) = e^{-(x-1)}, ~x\in \left[0,10\right]$, using the right-hand side (RHS) of equation (\ref{lin_comb_had}). The function is first plotted for $100$ points and then sampled at equally spaced $2^5$ points (5 qubit system). The function is then reconstructed using the with Hadamard basis functions and the sampled $2^5$ points. Fig.\ref{fig:basis_reconst} gives the actual as well as the reconstructed curves. 

\begin{figure}[htb!]
\centering 
\includegraphics[width=0.49\linewidth]{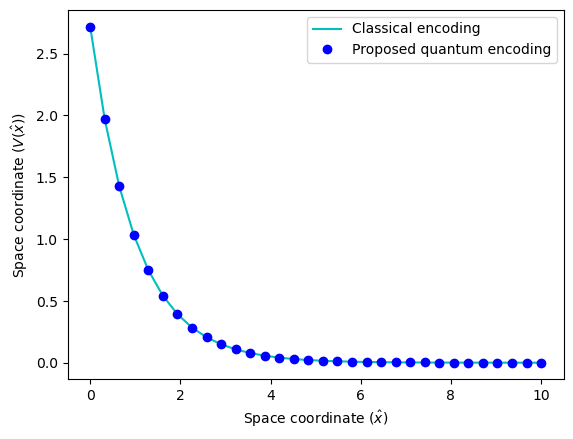} 
\caption{Actual curve and the reconstructed curve using sampled points}
\label{fig:basis_reconst}
\end{figure}
As can be seen from the figure, a perfect reconstruction is possible if we have $2^n$ basis functions.


\subsection{Quantum circuit for Hadamard basis encoding}

In our quantum circuit, we have to implement the operation given in (\ref{eqn:time_evol_pot}), i.e., construct the unitary operation given by $e^{-i\mathbf{V}(\hat{x})t}$. Now, the potential is encoded in the diagonal matrix. In order to reconstruct the diagonal matrix we need to have the basis functions in the diagonals. The $\mathbf{V}(\hat{x})$ matrix would then be a linear combination of matrices that have the Hadamard basis functions in the diagonals as shown in (\ref{eqn:diag_lin_comb}).
\begin{equation}
    \mathbf{V}(\hat{x}) = \sum_{j=1}^{2^n} c_jB_j \label{eqn:diag_lin_comb}
\end{equation}
For example, $B_1$ will be a matrix that has $b_1$ in the diagonals. The matrix $e^{-i\mathbf{V}(\hat{x})t}$ can then be expressed as follows.
\begin{equation}
    e^{-i\mathbf{V}(\hat{x})t} = e^{-it\sum c_jB_j} \label{eqn:exp_lin_comb}
\end{equation}
In (\ref{eqn:exp_lin_comb}), the index is $j$ and $i$ is the imaginary number $\sqrt{-1}$. Since diagonal matrices commute, the expression in (\ref{eqn:exp_lin_comb}) can be written as follows. In (\ref{eqn:exp_prod_comb}), $\theta _j = c_jt$.
\begin{equation}
    e^{-i\mathbf{V}(\hat{x})t} = \prod_{i=1}^{2^n}  e^{-it c_jB_j} = \prod_{i=1}^{2^n}  e^{-i\theta _jB_j} \label{eqn:exp_prod_comb}
\end{equation}

\subsection*{Expressing the basis diagonals using Pauli $Z$ and $I$ matrices}

It can be shown that the Hadamard basis functions can be put in diagonal matrices by using tensor product combinations of Pauli $Z$ and $I$ matrices. 
\begin{equation*}
    Z = \left[\begin{array}{rr}
        1 & 0 \\
        0 & -1
    \end{array}\right], ~~~~ I = \left[\begin{array}{rr}
        1 & 0 \\
        0 & 1
    \end{array}\right]
\end{equation*}

The combination of $Z$ and $I$ follows the binary number progression ($000...$ to $111...$, all zeros to all ones) with $I$ corresponding to $0$ and $Z$ corresponding to $1$. The leftmost column or the leftmost basis of the Hadamard functions matrix (refer to the matrix given in (\ref{eqn:Hadamard_func})) corresponds to all zeros and the rightmost column corresponds to all ones. Table \ref{table:Hadamard-IZ_comb} illustrates the idea for a $3$-qubit system.
\begin{table}[!h]
\caption{Hadamard basis matrices and Pauli $Z$ and $I$ equivalence}
\vspace{-0.4cm}
\begin{center}
\begin{tabular}{|c|c|c|c|}
 \hline
 Sl. No. & Basis matrix & Binary exp & $I$ and $Z$ combination \\
 \hline
 1 & $B_1$ & $000$ & $I\otimes I\otimes I$ \\
 2 & $B_2$ & $001$ & $I\otimes I\otimes Z$ \\
 3 & $B_3$ & $010$ & $I\otimes Z\otimes I$ \\
 4 & $B_4$ & $011$ & $I\otimes Z\otimes Z$ \\
 5 & $B_5$ & $100$ & $Z\otimes I\otimes I$ \\
 6 & $B_6$ & $101$ & $Z\otimes I\otimes Z$ \\
 7 & $B_7$ & $110$ & $Z\otimes Z\otimes I$ \\
 8 & $B_8$ & $111$ & $Z\otimes Z\otimes Z$ \\
 \hline
\end{tabular}\label{table:Hadamard-IZ_comb}
\end{center}
\end{table} 
$B_1$ is the diagonal matrix having the basis $b_1$ (the left most basis) in the diagonal. It corresponds to the binary number $000$. This matrix can then be constructed by taking the tensor product $I\otimes I\otimes I$. Similarly, for the rest, with the final basis matrix $B_8$ having the basis $b_8$ in the diagonal. This corresponds to $111$ in the binary number progression. This matrix can be created by doing the tensor product $Z\otimes Z\otimes Z$.

\subsection*{Construction of the quantum circuit with Hadamard encoding}

We have seen in (\ref{eqn:exp_prod_comb}) how $e^{-i\mathbf{V}(\hat{x})t}$ can be expressed as a product of unitaries,  $e^{-i\theta _jB_j}$, with $j=\{1, 2, \hdots, ~ 2^n\}$. Now, $B_j$ can be expressed in the form of $Z$ and $I$ matrices as given in Table \ref{table:Hadamard-IZ_comb}, scaled by a factor of $\frac{1}{\sqrt{N}}$. Expressions of the type $e^{-i\theta (I\otimes I\otimes Z)}$,  $e^{-i\theta (I\otimes Z\otimes I)}$ or $e^{-i\theta (Z\otimes I\otimes I)}$ can be implemented in a quantum circuit by simply having the gate $R_z (2\theta)$ in the qubit location corresponding to where $Z$ is located. For example, for $e^{-i\theta (Z\otimes I\otimes I)}$, we will need $R_z (2\theta)$ in the msqb. Expressions of the type $e^{-i\theta (I\otimes Z\otimes Z)}$ or $e^{-i\theta (Z\otimes Z\otimes I)}$ require entanglement gates and can be implemented using CNOT gates as shown in Fig. \ref{fig:z1xz2_imp} \cite{Mike_Ike}.

\begin{figure}[htb!]
\centering 
\includegraphics[width=0.49\linewidth]{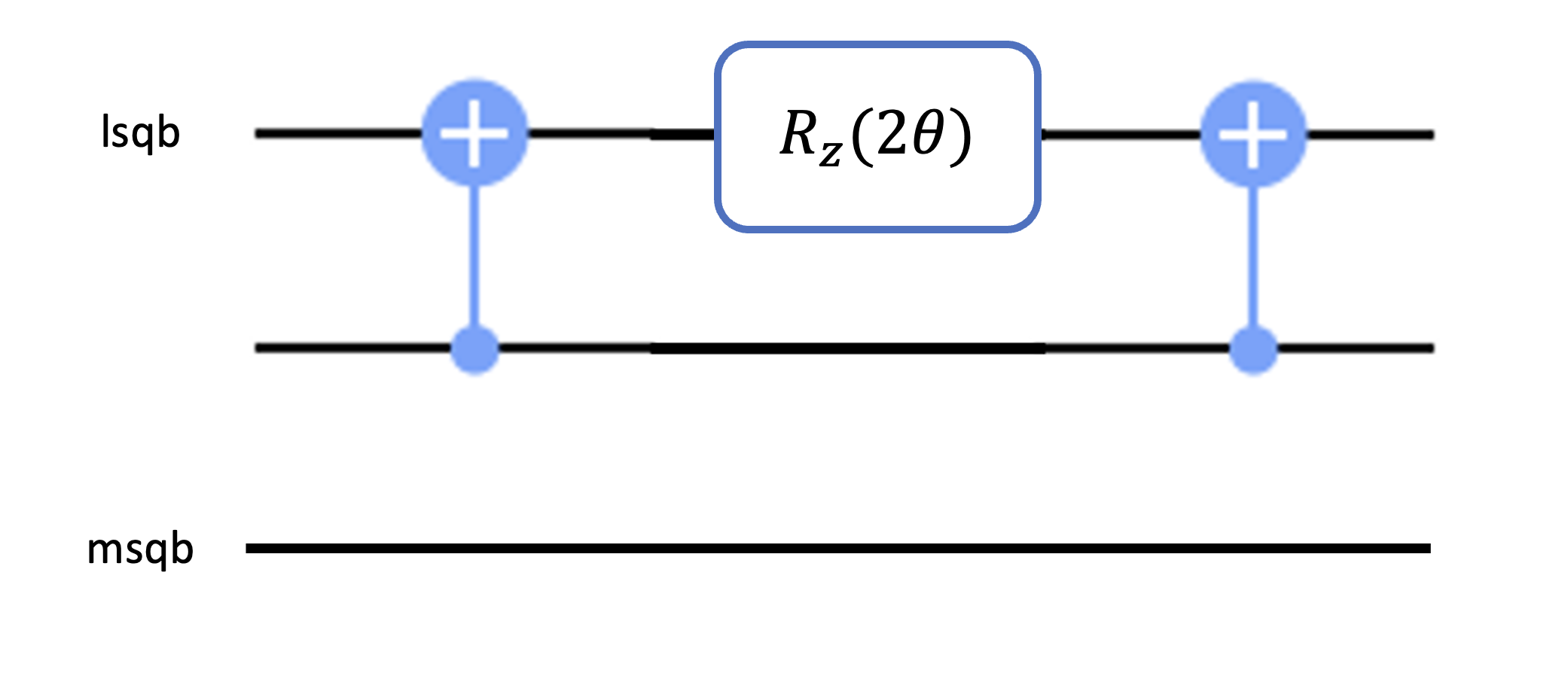} 
\caption{Implementation of $e^{-i\theta (I\otimes Z\otimes Z)}$}
\label{fig:z1xz2_imp}
\end{figure}
Similarly, the expression $e^{-i\theta (Z\otimes Z\otimes Z)}$ can be implemented using CNOT and $R_z$ gates as shown in Fig. \ref{fig:z1xz2xz3_imp}.
\begin{figure}[htb!]
\centering 
\includegraphics[width=0.49\linewidth]{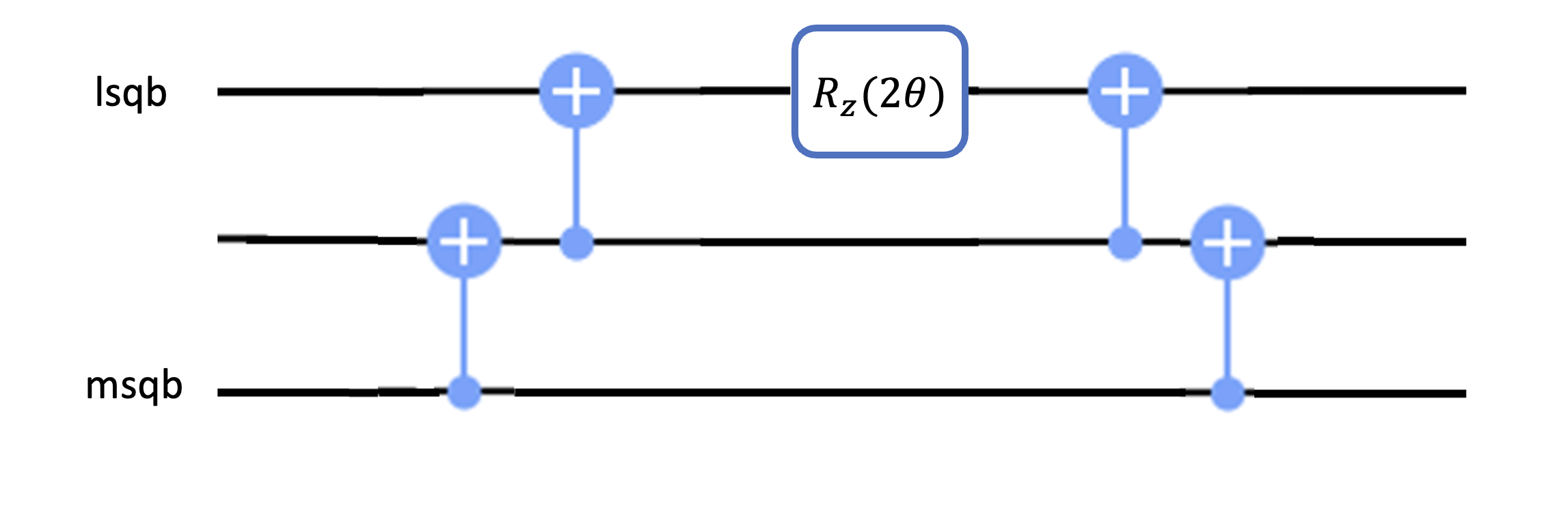} 
\caption{Implementation of $e^{-i\theta (Z\otimes Z\otimes Z)}$}
\label{fig:z1xz2xz3_imp}
\end{figure}
The expression $e^{-i\theta (I\otimes I\otimes I)}$ will result in a global phase of phase angle $-\theta$.

\subsection{Quantum polynomial approximate encoding}

We propose another Hamiltonian encoding procedure using a parameterized approach for an approximate time evolution with reduced computational complexity. In this procedure, we perform a polynomial fitting of the actual potential energy curve up to a desired accuracy ($\delta$) as follows
\begin{align}
    \underset{x,r}{\min}~\Vert e^{-i\mathbf{V}(\hat{x})t} - g(x,r) \Vert \leq \delta, \label{eqn:fit_tol}
\end{align}
where $g(x,r)$ is a polynomial in $x$ of degree $r$. 
To perfectly encode a function of degree $r$, using single and multi qubit phase gates in a quantum circuit, the total number of gates required will be given by the following expression.
\begin{align}
    N_g=^nC_0 + ^nC_1 + ^nC_2 + \dots + ^nC_r.
    \label{polyapprox}
\end{align}
In (\ref{polyapprox}), we have an $n$ qubit system and the number of possible multi-qubit phase gates involving $r$ qubits will be $^nC_r$ (for e.g., $^nC_1$ - single qubit gates, $^nC_2$ - two-qubit gates, $^nC_3$ - three-qubit gates, and so on). If we are to perfectly encode a function of degree $m$, the quantum system we will need is $n\geq m$ qubits. If we have $n=m$, the number of gates we will have in that case will be as follows.
\begin{align}
    N_g=^nC_0 + ^nC_1 + ^nC_2 + \dots + ^nC_n = 2^n.
    \label{eqn:polyapprox_exact}
\end{align}
The gate requirements are in agreement with the encoding strategy described by Grover in \cite{Grover_encoding} for creating integrable probability distributions.

Using this procedure, the constant term corresponding to $^nC_0$ can be encoded in the global phase. The linear coefficients corresponding to the $^nC_1$ terms require $n$ number of single-qubit gates, and the quadratic coefficients corresponding to the $^nC_2$ terms require two-qubit entangled gates of $^nC_2$ combinations, and so on till $^nC_r$ for $r^{th}$-order polynomial fitting. However, to reduce the complexity of the circuit, we can approximate the potential curve to a second or third-order polynomial and can choose single and two-qubit gates (and $3$ qubit entangled gates if required) to design a quantum circuit which can reconstruct the potential curve to a desired accuracy as given in (\ref{eqn:fit_tol}). The gate parameters may be estimated to do any of the two following tasks.
\begin{enumerate}
    \item Perfectly encode $K=^nC_0+ ^nC_1+ ^nC_2$ functional values with $K$ parameters and not worry about the remaining $2^n - K$ functional values. 
    \item Give a least squares solution for all the $2^n$ functional values with the above $K$ parameters, albeit, with some errors in all of them.
\end{enumerate}
We take the example of a least squares solution in a $3$ qubit quantum circuit, where a potential operator can be constructed. By placing $\theta_0$ as the global phase, phase gates ($P(\theta)$) with angles $\theta_1,\theta_2,\theta_3$ (respectively on the first, second and third qubits) and $2$-qubit controlled phase gates ($Cp(\theta)$) with phases $\theta_{12},\theta_{13},\theta_{23}$ (between qubit first and second, first and third, and second and third qubits, respectively) in the quantum registers, we will get a unitary operation that corresponds to a diagonal matrix $\tilde{\mathbf{U}}(\theta) \in \mathcal{C}^{8\times 8}$. $\tilde{\mathbf{U}}(\theta)$ should be a close approximation of the matrix exponential $e^{-i\mathbf{V}(\hat{x})t}$, with the matrix $\mathbf{V}(\hat{x})$ being a diagonal matrix with  elements $\mathbf{v}=[v_0,\dots, v_7]$ in its diagonal (as explained in in section \ref{sec:PE operator}). Here, the first qubit refers to the lsqb and the third qubit refers to the msqb.

We have already seen in Fig. \ref{fig:potential_states} and in Fig. \ref{NAI_potential} that the diagonal elements correspond to a particular state (lattice point), as $v_0$ corresponds to the state $\ket{000}$, $v_1$ to $\ket{001}$, etc. Let us now consider the following points.
\begin{itemize}
    \item Every time we add a phase gate, $P(\theta)$, to a qubit, all states that have that qubit in state $\ket{1}$ get affected by an imaginary exponential of that phase parameter. For example, the phase gate parameter $\theta_1$ will affect all the states where the first qubit (lsqb) is in state $\ket{1}$ - $e^{i\theta_1}\ket{001}, e^{i\theta_1}\ket{011}, e^{i\theta_1}\ket{101}$ and $e^{i\theta_1}\ket{111}$.

    \item Every time we add a controlled phase gate or a two-qubit phase gate, $Cp(\theta)$, all states that have those qubits in state $\ket{1}$ get affected. For example, the phase gate $\theta_{12}$ will affect all the states where the first qubit and the second qubit are in state $\ket{1}$ - $e^{i\theta_{12}}\ket{011}$ and $e^{i\theta_{12}}\ket{111}$.
\end{itemize}
Putting all the $P(\theta)$ and $Cp(\theta)$ gates together along with the global phase $\theta_0$, if we compare the exponential in the diagonal matrix $\tilde{\mathbf{U}}(\theta)$ and the matrix exponential $e^{-i\mathbf{V}(\hat{x})} (t=1)$, we should get the following equations in matrix form with $K$ ($^3C_0+ ^3C_1+ ^3C_2 = 7$) parameters.
\begin{align}
    \begin{bmatrix}
  1&  0&  0&  0&  0&  0&  0\\
  1&  1&  0&  0&  0&  0&  0\\
  1&  0&  1&  0&  0&  0&  0\\
  1&  1&  1&  0&  1&  0&  0\\
  1&  0&  0&  1&  0&  0&  0\\
  1&  1&  0&  1&  0&  1&  0\\
  1&  0&  1&  1&  0&  0&  1\\
  1&  1&  1&  1&  1&  1&  1
    \end{bmatrix}\begin{bmatrix}
        \theta_0\\ \theta_1\\ \theta_2 \\ \theta_3 \\ \theta_{12} \\ \theta_{13} \\ \theta_{23}
    \end{bmatrix} = \begin{bmatrix}
        v_0\\ v_1\\ v_2 \\ v_3 \\ v_4 \\ v_5 \\ v_6 \\ v_7
    \end{bmatrix},
    \label{ls3eqn}
\end{align}
Equation (\ref{ls3eqn}) has the form $\mathbf{A} \boldsymbol{\xi} = \mathbf{v}$, with $\mathbf{A}\in \mathbb{R}^{8\times 7}$ being the matrix, $\boldsymbol{\xi}\in \mathbb{R}^{7\times 1}$ denotes the parameter vector, and $\mathbf{v}\in \mathbb{R}^{8\times 1}$ represents the vector with functional values corresponding to the potential energy curve. 

More generally, with a finite number of $n$ qubits, the dimension of $\mathbf{A}$ will be $N\times K$ and $\mathbf{v}\in \mathbb{R}^{N}$ for $K=^nC_0+ ^nC_1 + ^nC_2+\dots + ^nC_r$ parameters (for $r^{th}$ order polynomial approximation).

One can perform a least square estimate of the parameter vector $\boldsymbol{\xi}$ as
\begin{align}
    \underset{\theta}{\min} ~\Vert \mathbf{A}\boldsymbol{\xi} - \mathbf{v} \Vert_2 
    \label{lsmin}
\end{align}
where, $\boldsymbol{\xi} \in \mathbb{R}^K$.
The least square solution for (\ref{lsmin}) can obtained as $\hat{\boldsymbol{\xi}}=\mathbf{A}^{+} \mathbf{v}$, where $\mathbf{A}^{+}$ is the pseudo-inverse of $\mathbf{A}$. Note that, while solving for the least square problem, one can rearrange the columns and rows of the matrix $\mathbf{A}$ (and accordingly in the equation (\ref{ls3eqn})) to make it in a lower triangular matrix form. For example, (\ref{ls3eqn}) can be written in the lower-triangular form as follows
\begin{align}
    \begin{bmatrix}
  1&  0&  0&  0&  0&  0&  0\\
  1&  1&  0&  0&  0&  0&  0\\
  1&  0&  1&  0&  0&  0&  0\\
  1&  0&  0&  1&  0&  0&  0\\
  1&  1&  1&  0&  1&  0&  0\\
  1&  1&  0&  1&  0&  1&  0\\
  1&  0&  1&  1&  0&  0&  1\\
  1&  1&  1&  1&  1&  1&  1
    \end{bmatrix}\begin{bmatrix}
        \theta_0\\ \theta_1\\ \theta_2 \\ \theta_3 \\ \theta_{12} \\ \theta_{13} \\ \theta_{23}
    \end{bmatrix} = \begin{bmatrix}
        v_0\\ v_1\\ v_2 \\ v_4 \\ v_3 \\ v_5 \\ v_6 \\ v_7
    \end{bmatrix},
    \label{ls3eqn1}
\end{align}
By doing this operation, parameters can be estimated efficiently.

\subsection{Circuit construction with polynomial approximation}

\begin{figure}[htb!]
\centering
\subfloat[One input phase gate (with phase $\theta_1$) to implement $\mathbf{I}\otimes \mathbf{I}\otimes e^{i\theta_1}$]{
 \includegraphics[width=0.33\linewidth]{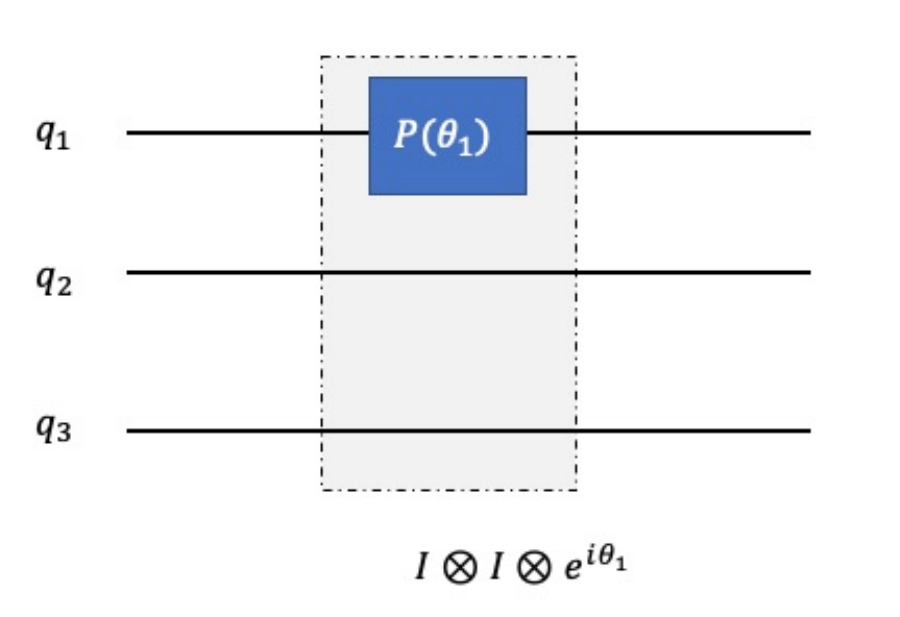} 
}
\subfloat[Two input controlled phase gate (with phase $\theta_{12}$) to implement $\mathbf{I}\otimes Cp $]{
  \includegraphics[width=0.33\linewidth]{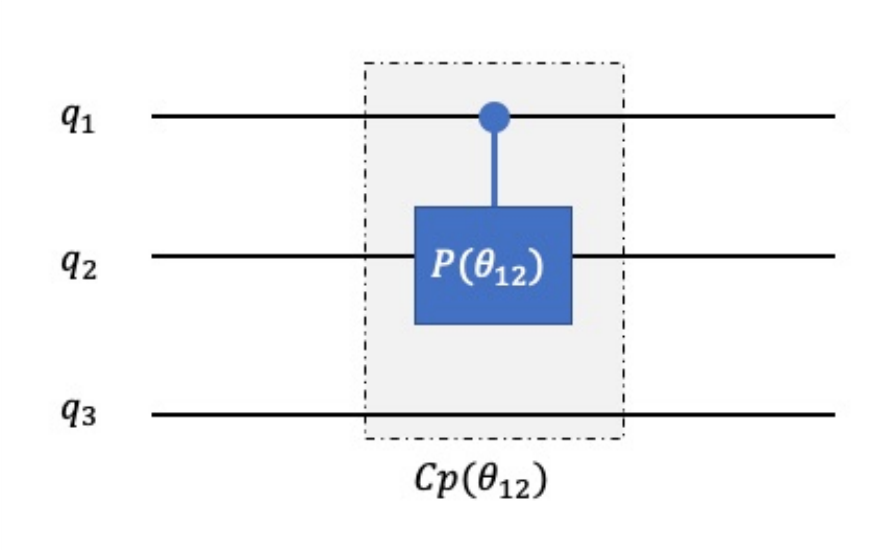} 
}
\subfloat[Three qubit gate to implement the entangled operator $CCp$ with phase $\theta_{123}$.]{
  \includegraphics[width=0.33\linewidth]{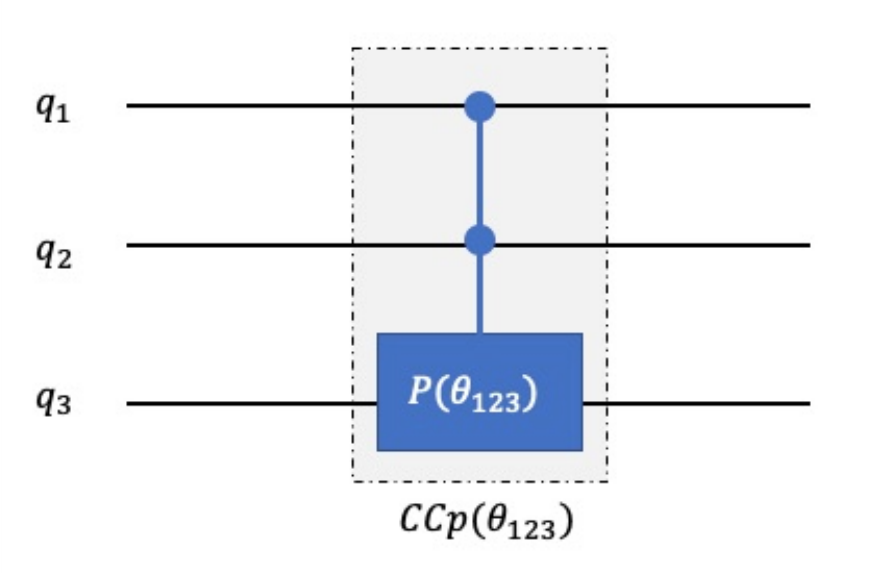} 
}
\caption{Circuit construction with polynomial approximation encoding in a $3$ qubit circuit: Here, the input qubits are denoted by $q_1,q_2$, and $q_3$}
\label{poly_approx_ckt}
\end{figure}
With polynomial approximation encoding, we seek to implement quantum circuits that can realize an approximate function (with a finite order) closer to the potential energy to be represented in the diagonal matrix form. We will be selecting the phase gates to prepare the diagonal matrix. Here, $1$-qubit phase gates can approximate the potential energy function to the first order. Similarly, two input-controlled phase gates ($Cp$) and three input gates ($CCp$) can be used to approximate up to $2^{nd}$, and $3^{rd}$ order polynomial of the potential energy function. Here, in Fig.\ref{poly_approx_ckt}, we have shown the circuit implementation with different gates ($1$-input, $2$-input, and $3$-input gates) to prepare a diagonal matrix with desired phase angles. The phase and entangled phase gates used in this figure are defined as follows
\begin{align}
    P(\theta_1)&=\begin{bmatrix}
        1&0\\0&e^{i\theta_1}
    \end{bmatrix},\\
    Cp(\theta_{12})&= \mathbf{I}\otimes \ket{0}\bra{0}~ + ~P(\theta_{12})\otimes\ket{1}\bra{1}\nonumber\\
    &= \begin{bmatrix}
        1&0&0&0\\
        0&1&0&0\\
        0&0&1&0\\
        0&0&0&e^{i\theta_{12}}
    \end{bmatrix},\\
    CCp(\theta_{123})&= \mathbf{I}\otimes \mathbf{I}\otimes \ket{0}\bra{0}~ + ~Cp(\theta_{123})\otimes\ket{1}\bra{1}\nonumber\\
    &=\begin{bmatrix}
        1&0&0&0&0&0&0&0\\
        0&1&0&0&0&0&0&0\\
        0&0&1&0&0&0&0&0\\
        0&0&0&1&0&0&0&0\\
        0&0&0&0&1&0&0&0\\
        0&0&0&0&0&1&0&0\\
        0&0&0&0&0&0&1&0\\
        0&0&0&0&0&0&0&e^{\theta_{123}}
    \end{bmatrix}\label{ccp_eqn}.
\end{align}

As examples, we have shown that placing a phase gate with angle $\theta_1$ can be used to implement an operator $\mathbf{I}\otimes \mathbf{I} \otimes e^{i\theta_1}$ for a $3$ qubit quantum circuit shown in Fig.\ref{poly_approx_ckt}.$(a)$. The unitary operator obtained by placing a phase gate ($P(\theta_1)$) in the first qubit is given by
\begin{align}
    \mathbf{U}_1 = \begin{bmatrix}
        1&0&0&0&0&0&0&0\\
        0&e^{i\theta_1}&0&0&0&0&0&0\\
        0&0&1&0&0&0&0&0\\
        0&0&0&e^{i\theta_1}&0&0&0&0\\
        0&0&0&0&1&0&0&0\\
        0&0&0&0&0&e^{i\theta_1}&0&0\\
        0&0&0&0&0&0&1&0\\
        0&0&0&0&0&0&0&e^{i\theta_1}
    \end{bmatrix}.
\end{align}
The placing of this unitary operator will lead to phase operations $e^{i\theta_1}\ket{001}, e^{i\theta_1}\ket{011}, e^{i\theta_1}\ket{101}$ and $e^{i\theta_1}\ket{111}$, as explained in the previous section.
Similarly, one can get the operator $\mathbf{I}\otimes P(\theta_2)\otimes \mathbf{I}$ by placing a phase gate $P(\theta_2)$ with angle $\theta_2$ in the second register, and so on till the last qubit. We can place $n$ phase gates to an $n$-qubit circuit to approximate a first-order polynomial for the potential energy function. 

For the second-order approximation, we need $2$-qubit quantum gates. Here, we have chosen controlled phase gates ($Cp$) for the same. We can place a controlled phase gate between any two qubits in a $n$ qubit system to generate a unitary matrix. For a $n$ qubit quantum system, there are $^nC_2$ possible combinations to place different $Cp$ gates in the circuit. Similarly, we can use a $CCp$ gate to prepare a unitary of the form given in (\ref{ccp_eqn}). By placing phase gates (with angles $\theta_1, \theta_2, \theta_3$) and controlled phase gates (with angles $\theta_{12},\theta_{13},\theta_{23}$), we can get a linear combination of angles in the diagonal of the unitary matrix as shown in (\ref{ls3eqn}). Here, by solving the linear equations in the least square sense, we have obtained the estimated angles (parameters) to prepare the second-order polynomial of the potential energy function.

\section{Gate Complexity} \label{sec:gate_complexity}

In this research work, we have shown two different approaches for encoding the (discretized) potential energy function as a diagonal Hamiltonian operator. The complexity of this encoding method for the Hamiltonian simulation is measured in terms of gate complexity represented as a function of the number of qubits ($n$) as follows.

\subsection{Hadamard Encoding}

In the Hadamard encoding technique, CNOT and $Rz$ gates are used for the realization of the potential energy function. The number of gates required for simulating a potential operator $\mathbf{V}(\hat{x})\in \mathbb{R}^{2^n \times 2^n}$ using an $n$-qubit circuit is given in the below lemma.

\begin{lemma}
    A potential energy function can be encoded for the Hamiltonian simulation in an $n$ qubit circuit with $\sum_{r=2}^{n} {}^nC_2 2(r-1)$ number of quantum gates using the Hadamard basis encoding.
\end{lemma}

\begin{proof}
We need in total $2^n$ operations. If we discount the global phase, that would result in $2^{n} - 1$ phase gates (single qubit gates). If we look at Table \ref{table:Hadamard-IZ_comb}, we will see that there are $^nC_1$ number of tensor products having one Pauli $Z$, $^nC_2$ number of tensor products having two Pauli $Z$s, and so on with $^nC_r$ number of tensor products having $r$ Pauli $Z$s. Every time we have $r$ Pauli $Z$ matrices in a tensor product, we will need to use $2(r-1)$ CNOT gates. The total number of CNOT gates required will be as given in (\ref{eqn:n_cnot}).
\begin{equation}
    n_{CNOT} = \sum_{r=2}^{n} {}^nC_r 2(r-1). \label{eqn:n_cnot}
\end{equation}
\end{proof}
For a $4$ qubit system, one can see from Fig. \ref{4qubitcircuit}.$(a)$ that there are $2^4 - 1 = 15$ phase (single qubit) gates. The number of CNOT gates required is given by
\begin{flalign*}
    n_{CNOT} &= {^4C_2}2(2-1) + {^4C_3}2(3-1) + {^4C_4}2(4-1) \\
             &=(6\times2\times1) + (4\times2\times2) + (1\times2\times3) = 34. 
\end{flalign*}

The Hadamard encoding technique can be used for simulating any arbitrary potential energy function with very high precision as it uses $2^n$ basis functions. However, the total count of CNOT gates is very high. \cite{welch2014efficient} gives a method using Gray code scheme to reduce the number of CNOT gates. One can also find other compaction techniques to do CNOT gate optimization \cite{maslov-ckt}. The approximate polynomial encoding technique reduces gate complexity with a trade-off against precision. The complexity versus precision trade-off is based on the order of the polynomial fitting.

\subsection{Approximate polynomial encoding}

Using the polynomial approximation of the actual potential energy function (with a reliable order $r$ of the polynomial), the gate complexity is presented in the below proposed lemma. 

\begin{lemma}
    The potential energy function can be encoded in an $n$-qubit circuit with the least square encoding method using at most $\Theta(^nC_r)$ gate complexity for a $r^{th}$-order polynomial approximation of the potential energy. With second-order approximation, we need $n+ ^nC_2$ number of quantum gates for the Hamiltonian simulation. 
\end{lemma}

\begin{proof}
With the discretization process, we take $N$-sampled values of the potential energy function encoded in $n=log(N)$ qubits ($N$ is taken in the power of $2$). With our proposed least square method, we can take a polynomial approximation of the function to order $r$ as given in (\ref{polyapprox}). Starting the global phase with the term $^nC_0$, the other terms $^nC_r$ denote the combination of quantum gates with at most $r^th$-input quantum gate.  Hence, the combinations of the total gate counts in the quantum circuit is given by 
\begin{align}
    n_T=^nC_1 + ^nC_2 +\dots + ^nC_r,
\end{align}
which is in $\Theta(^nC_r)$. Restricting the approximation up to the order of $r=2$, we need $n$ phase gates, and $^nC_2$ number of controlled phase gates (at most) for the least square approximation. Hence, the complexity becomes $n+ {}^nC_2$. 
\end{proof}

Note that, the computational gate complexity using a second-order polynomial fitting to the potential energy curve is reduced to $n+ {}^nC_2 \approx \boldsymbol{\Theta}(n^2)$. This is a significant reduction of the computational resources in terms of the basic elementary quantum gates.





\section{Results} \label{sec:results}


In this result section, we demonstrate numerical simulation results for the time evolution operator design and potential energy reconstruction performed on IBM quantum simulators using qiskit \cite{qiskit_1},\cite{qiskit_2}. In the below subsections, we show the implementation of the potential energy operator for $3$ and $4$ qubit system as examples using proposed Hadamard basis encoding and polynomial approximate encoding methods. The performance of the proposed framework of the time evolution operator is measured with gate complexity as the key parameter index (KPI). In Table-\ref{table_param}, we have shown the parameters taken in the simulation environment.

\begin{table}[htb!]
    \centering
\begin{tabular}{|c|c|}
    \hline
    Quantum Simulator &  Unitary  Simulator, Qasm Simulator \\
    \hline
    IBMQ Machine & ibmq\_mumbai, ibm\_nairobi \cite{ibm-quantum}\\
    \hline
    Number of qubits ($n$) & $3-10$\\
    \hline
    Number of shots & $10000$\\
    \hline
    Range of space coordinate ($x$) & $[0, ~10]$\\
    \hline
    Sampling interval ($dx$) & $10/{2^n}$\\
    \hline
    Potential energy curve & $V(x)= a_1e^{-a_2(x-r_1)}$\\
    \hline
    Example & Potential energy of NaI\\
    \hline 
\end{tabular}
    \caption{Simulation parameters}
    \label{table_param}
\end{table}


\subsection{Encoding potential energy with $4$ qubit quantum circuit using Hadamard basis encoding and $2^{nd}$ order polynomial approximation}

Fig. \ref{4qubitcircuit}.$(a)$ gives the quantum circuit for a $4$ qubit system as implemented in qiskit. With $4$ qubit-based circuit we have used amplitude encoding to encode the sampled potential energy function $\mathbf{V}(\hat{x})\in \mathbb{R}^{16\times 1}$ with computational bases $\ket{0000}$ to $\ket{1111}$. The term with the basis $\mathbf{I}^{\otimes 4}$ is encoded here as the global phase. The number of basis states prepared with the combination of Pauli-$Z$ operators is $2^4=16$.

To encode with Hadamard basis as given in \ref{hadamardbasis_theory}, we have used the combination of $Rz$ gates and CNOT gates as shown. For encoding in Hadamard basis for $4$-qubits, one would need $15$ number of single qubit gates ($Rz$), and $34$ two qubit gates (CNOT). However, on further optimization the number of CNOT gates can be reduced \cite{maslov-ckt}. Fig. \ref{4qubitcircuit} gives a circuit which requires $30$ CNOT gates. The circuit we have provided here is a generic circuit upon which further circuit optimizations are possible as shown in \cite{welch2014efficient}. The number of CNOT gates given in (\ref{eqn:n_cnot}) can be considered as an upper bound on the number of two qubit gates required. The quantum Hadamard encoded unitary matrix $\mathbf{U}_{PE}$ approximates the function arbitrarily close to the classical function. We have performed a quantum simulation with evolution time $dt=1$, and potential energy function $\mathbf{V}(\hat{x})=e^{1-x}$ on qiskit unitary simulator, and the corresponding plot of the diagonal vector is given in Fig. \ref{construct_reconstruct}. 

\begin{figure}[htb!]
\centering
\subfloat[Circuit implementation in qiskit of a $4$ qubit system with Hadamard basis encoding]{
 \includegraphics[width=0.95\linewidth]{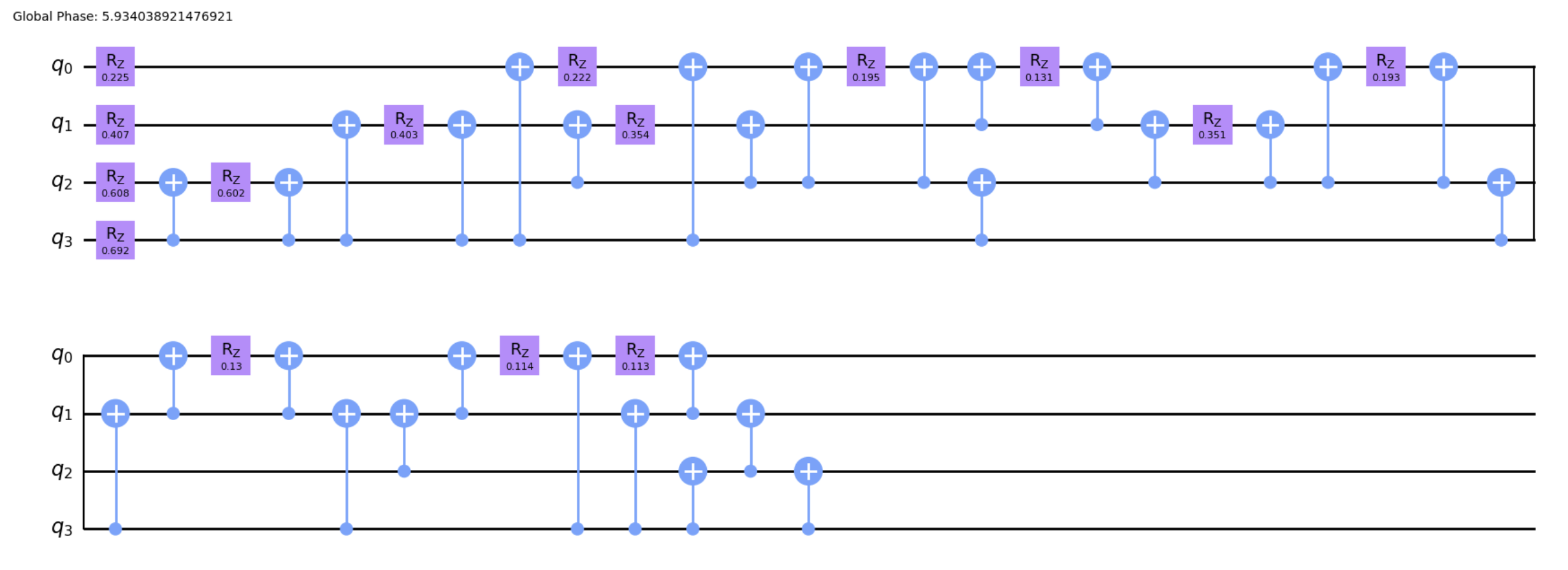} 
}
\hspace{0mm}
\subfloat[Circuit implementation in qiskit of a $4$ qubit system with $2^{nd}$ order polynomial approximation based least square encoding]{
\includegraphics[width=0.95\linewidth]{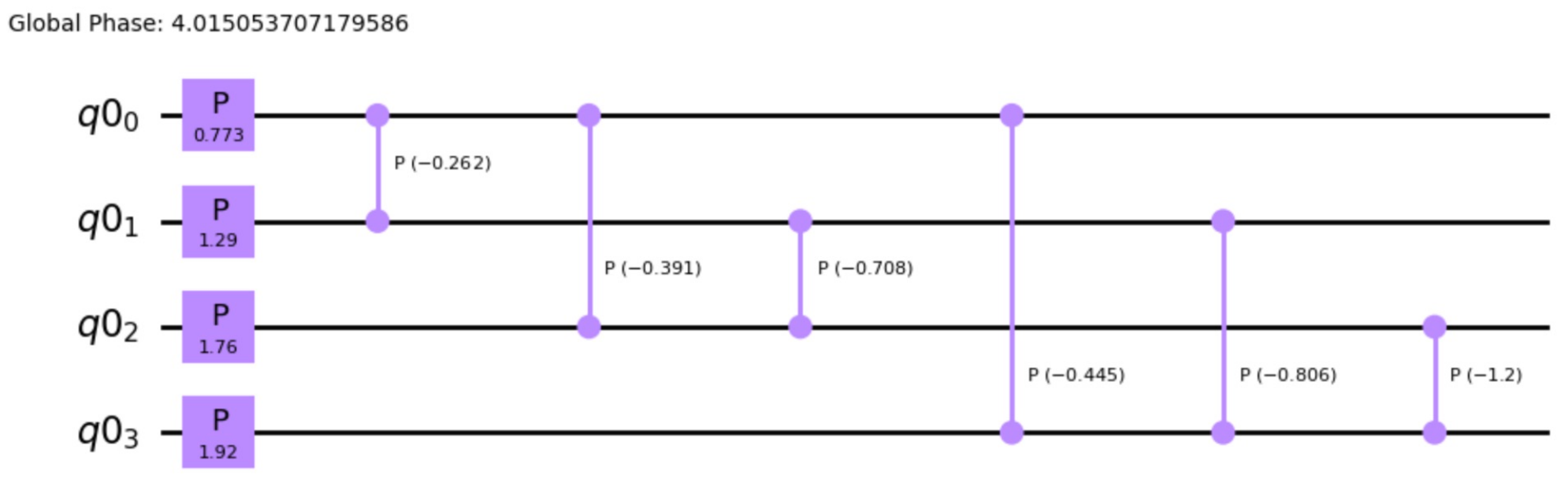} 
}
\caption{Time evolution operator for the potential energy encoding using $4$ qubit quantum circuit}
\label{4qubitcircuit}
\end{figure}

In Fig. \ref{4qubitcircuit}.$(b)$, we show a $4$-qubit quantum circuit with the $2^{nd}$-order approximate polynomial encoding. 
Here, we have used four $1$-qubit phase gates and six $ 2$-qubit controlled phase gates. Hence, the total number of quantum gates used is $10$. Here, the time evolution operator (corresponding to the potential energy), i.e., $e^{-i\mathbf{V}(\hat{x})t}$ is approximated with $2^{nd}$-order polynomial approximation function $g(x,2)$ with $10$ parameters (phases of the quantum gates). Starting the first potential energy sample (i.e, $v_0 \subset \mathbf{V}(\hat{x})$) encoded as a global phase, we have a total $11$ parameters for the least square estimate of the parameter vector $\hat{\boldsymbol{\xi}}$.  

Considering an exponentially decaying potential energy function of the form $\mathbf{V}(\hat{x})=e^{(1-x)}$ is simulated with Hadamard basis encoding and polynomial approximation encoding method as shown in Fig. \ref{construct_reconstruct}. In Fig. \ref{construct_reconstruct}.$(a)$, the plot of the diagonal unitary operator $U_{PE}=e^{-i\mathbf{V}t}$ is portrayed (imaginary values of the plot are shown here). The potential energy is discretized to $16$ samples ($v_0,\dots, v_{15}$) and encoded in $4$ qubit quantum circuit. The Hadamard basis encoding method constructs the unitary operator arbitrarily closer to the classically encoded values ($e^{-iv_0 t}, \dots, e^{iv_{15}t}$). The approximate encoding with the least square method (for $2^{nd}$ order polynomial approximation of the potential energy function) gives a closer approximation of the classically encoded potential energy function. 

In Fig. \ref{construct_reconstruct}.$(b)$, we have shown the reconstruction plot of the potential energy function ($V(\hat{x})$) obtained from the quantum simulated unitary operator. Here, we see a similar signature as seemed in the unitary construction. Here, the unitary construction of the potential energy using the least square encoding is slightly degraded as compared to the Hadamard basis encoding. However, one can simulate the polynomial approximation on a real quantum computer with fewer computational resources. As an example, for a $4$ qubit quantum circuit, we have used $10$ quantum gates for the $2^{nd}$ order polynomial approximation encoding of the potential energy function, whereas, Hadamard basis encoding takes $49$ quantum gates.

\begin{figure}[htb!]
\centering
\subfloat[Potential energy encoding in the Time evolution operator for $4$ qubit system and $2^{nd}$ order approximation: Imaginary part of $U_{PE}$ in vertical axis, discretized space ($x$) in horizontal axis]{
 \includegraphics[width=0.49\linewidth]{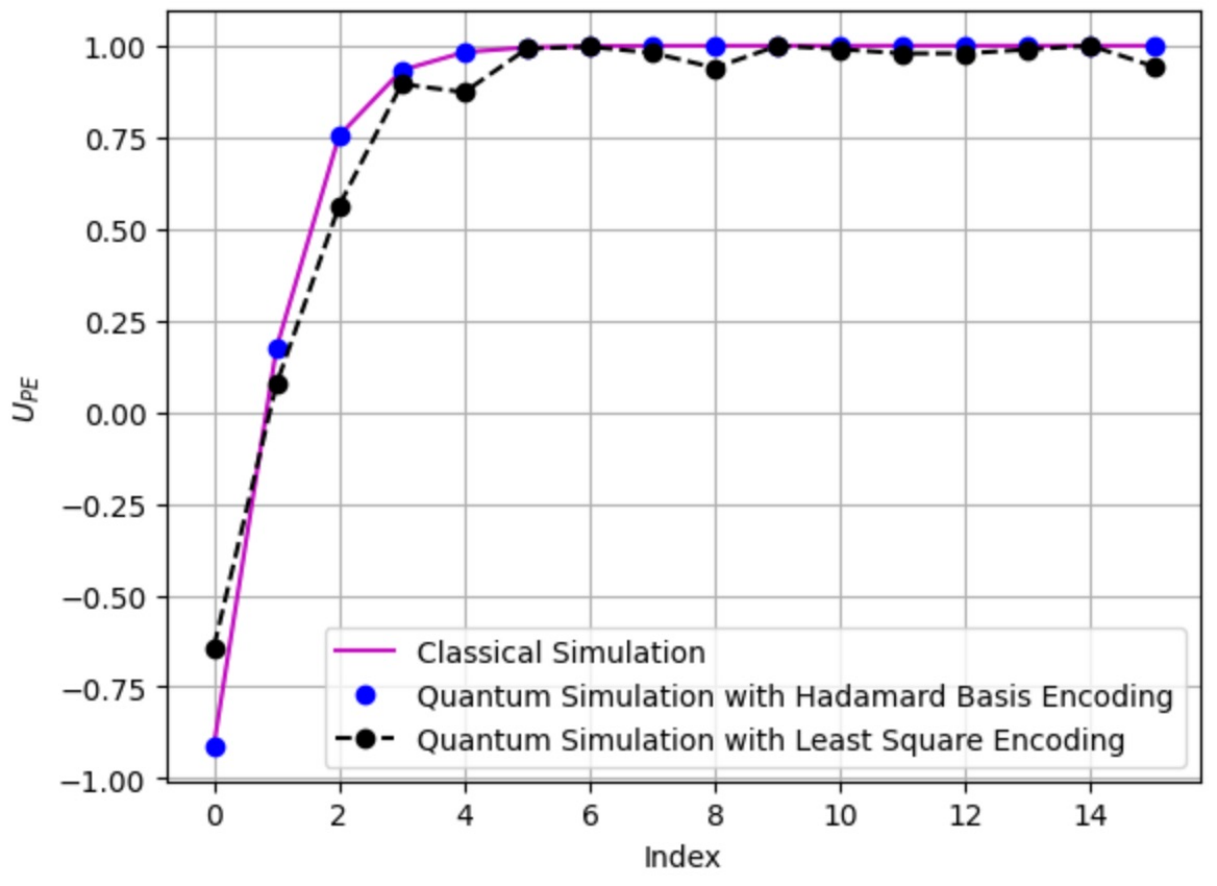} 
}
\subfloat[Potential energy reconstruction for $4$ qubit system and $2^{nd}$ order approximation]{
  \includegraphics[width=0.49\linewidth]{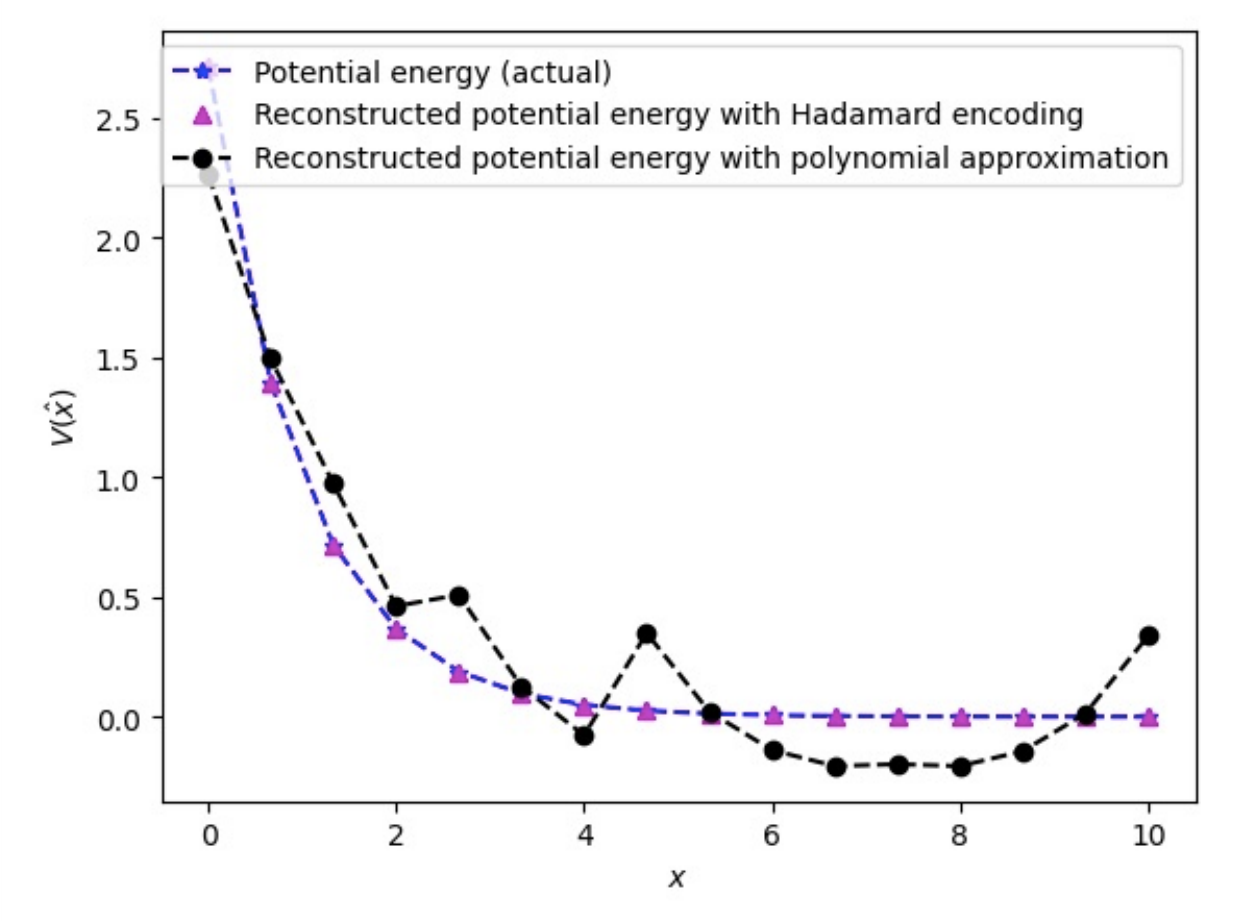} 
}
\caption{Potential energy encoding and reconstruction using Hadamard basis encoding and polynomial approximation encoding}
\label{construct_reconstruct}
\end{figure}

\subsection{ Potential energy evolution with $3^{rd}$-order polynomial approximation using a $3$-qubit quantum system}
\begin{figure}[htb!]
\centering
\subfloat[Potential energy encoding for $3^{rd}$ order polynomial approximation with $3$ qubit quantum circuit: Here a CCp gate is used with phase $0.190$, in addition to $3$ phase and $3$ controlled phase gates]{
\includegraphics[width=0.9\linewidth]{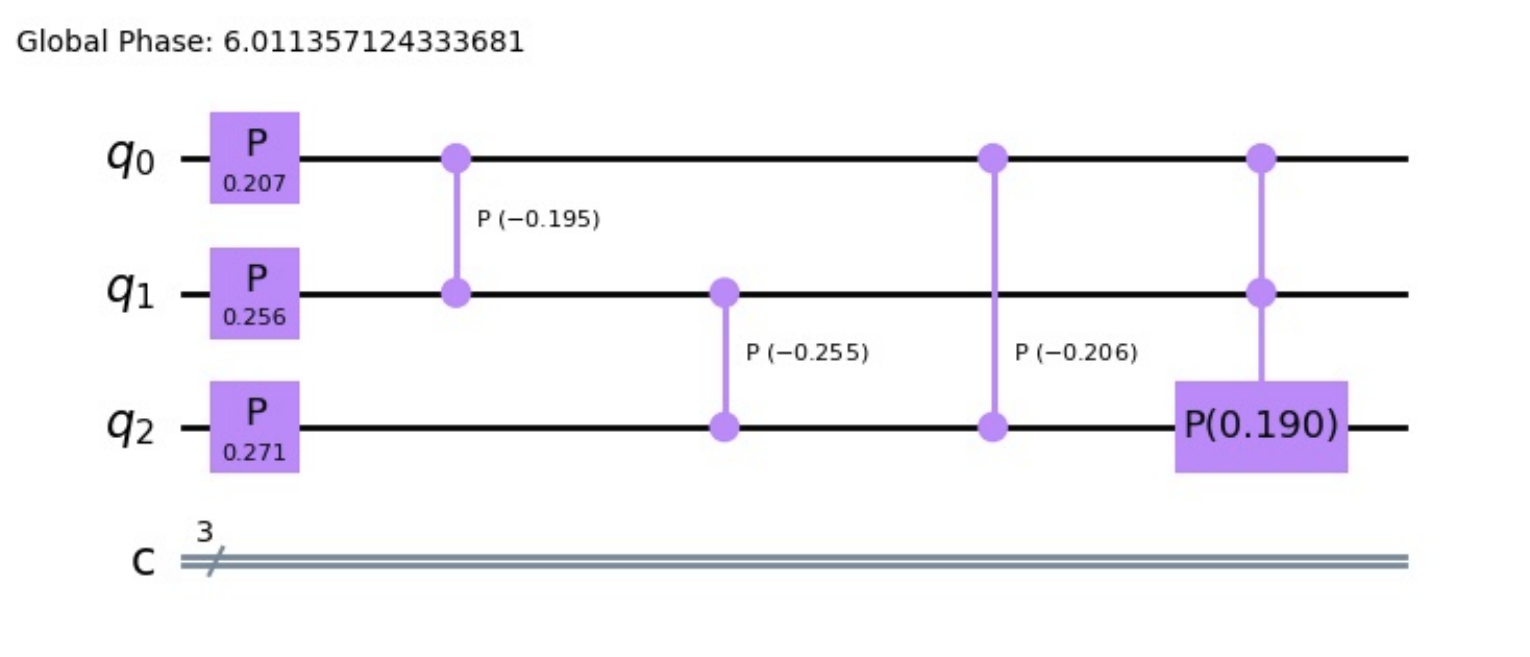} 
}
\hspace{0mm}
\subfloat[Potential energy encoding with $3^{rd}$ order polynomial approximation in $3$ qubit circuit: Imaginary part of $U_{PE}$ in vertical axis, discretized space ($x$) in horizontal axis]{
\includegraphics[width=0.49\linewidth]{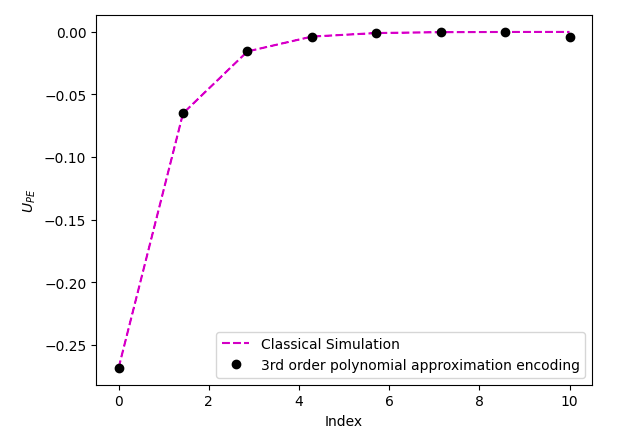} 
}
\subfloat[Potential energy reconstruction with $3^{rd}$ order polynomial approximation in $3$ qubit circuit]{
\includegraphics[width=0.49\linewidth]{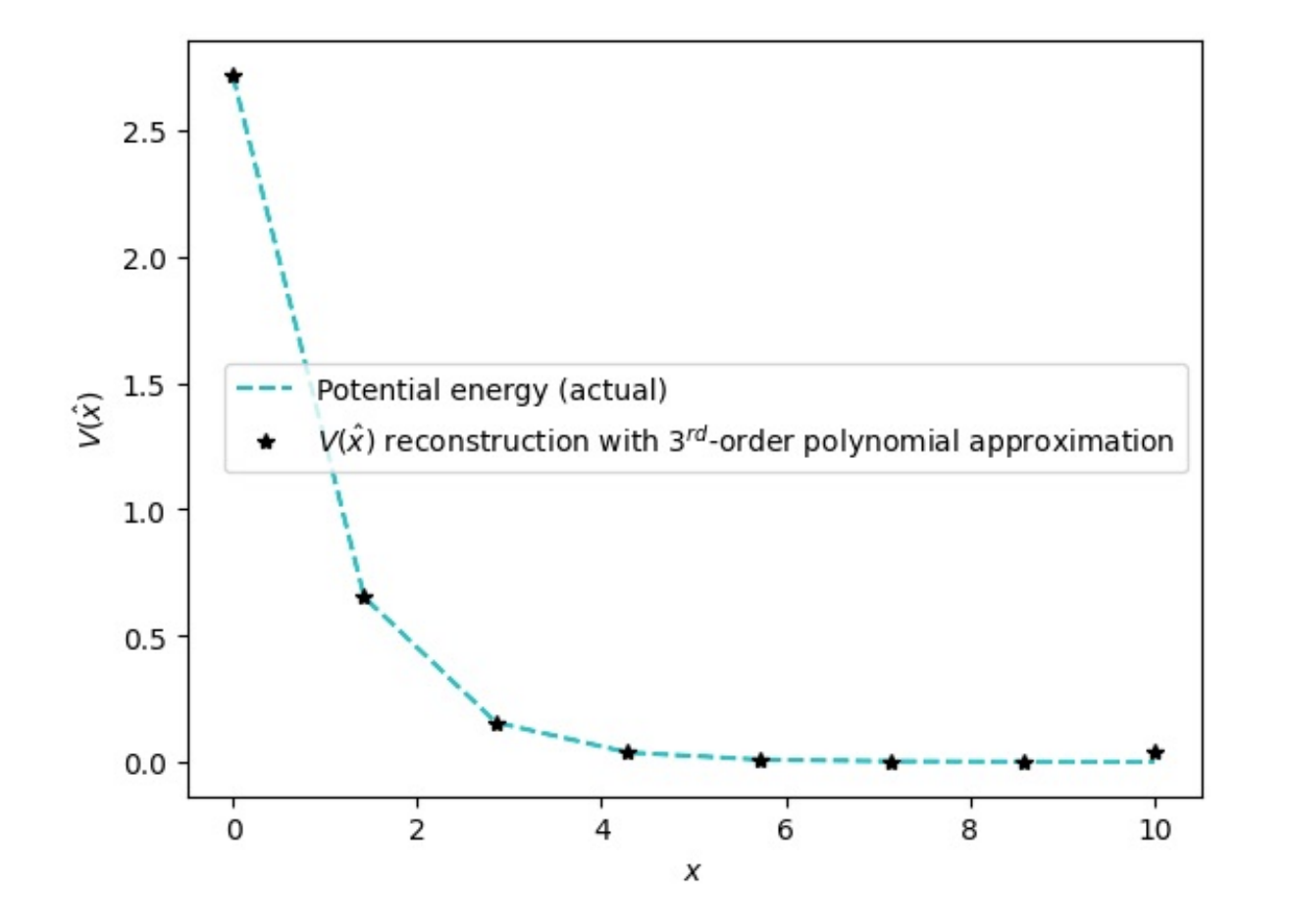} 
}
\caption{Potential energy encoding and reconstruction using polynomial encoding with $3^{rd}$ approximation}
\label{construct_reconstruct1}
\end{figure}

To improve the performance of the least square encoding of the potential energy, one can take higher-order polynomial approximation. For example, $3$-qubit quantum gates (e.g., controlled-controlled phase (CCp) gate) can be used in addition to phase and controlled phase gates for a $3^{rd}$ order polynomial approximation of the potential energy. An example is shown in Fig. \ref{construct_reconstruct1}.$(a)$ for a $3$-qubit quantum circuit. Here, we have used $3$ phase gates, $3$ controlled phase gates in an entangled state, and $1$ CCp gate. As a consequence, we have $8$ parameters (including $1$ global phase parameter) for the least square. The construction of the CCp gate may consider the composition of the unitary gate and CNOT gates (detailed in page. $182$ in \cite{chuang2002quantum}).  Using this $3^{rd}$ order polynomial approximation construction of the unitary evolution operator $U_{PE}$ is shown in Fig.\ref{construct_reconstruct1}.$(b)$, and the reconstruction of the potential energy $V(\hat{x})$ is shown in Fig.\ref{construct_reconstruct1}.$(c)$. The study of $2$-input gate-based encoding in Fig.\ref{construct_reconstruct}, and the $3$-input gate-based approximate encoding in Fig.\ref{construct_reconstruct1} shows that, with additional cost of higher-input quantum gate, the approximation accuracy can be improved. As higher input quantum gates are severely noisy in practical quantum computers presently, we focus circuit implementation primarily using $1$, and $2$ qubit gates with improved computational gate complexity for practical advantages.

\subsection{Example: Time evolution of the potential energy in NaI}

A prime example often chosen in the femtochemistry literature \cite{baskin2001freezing} to study the wave packet motion due in the presence of potential energy is Sodium Iodide (NaI) molecule. 
 The analytical expression of the potential energy ($V(\hat{x})$) of NaI as a function of atomic distance ($x$) is given by
\begin{align}
    V(x)= a_1 e^{-a_2 (x-r_1)}
\end{align}
with $a_1=0.0299$, $a_2=2.163$, and $r_1=5.102$ respectively. 

\subsubsection{Reconstruction with $2^{nd}$ order polynomial}

\begin{figure}[htb!]
    \centering
    \includegraphics[width=0.60\linewidth]{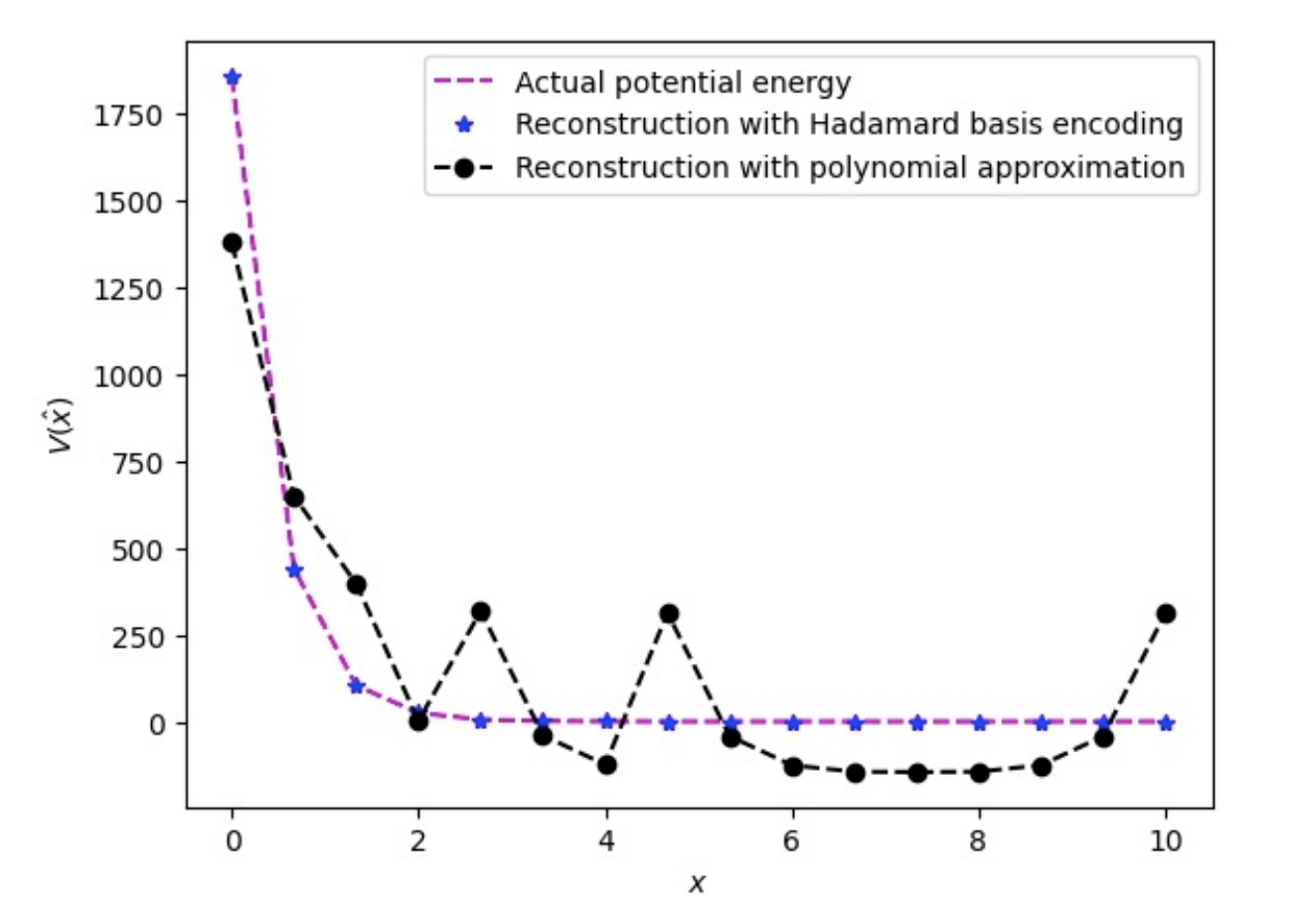}
    \caption{Potential energy reconstruction with Hadamard basis encoding and polynomial approximation encoding}
    \label{figreconfig}
\end{figure}

Here, in Fig. \ref{figreconfig}, we show the reconstruction of the potential energy function with two proposed techniques, i.e., Hadamard basis encoding and the polynomial approximation encoding. Here, we have used $4$ qubits to encode potential energy functions on a quantum computer. The Hadamard basis encoding circuit reconstructs the potential energy with high closeness to the actual (classical) potential energy curve. However, it takes $49$ quantum gates to encode  the potential energy function. On the other side, we have seen that a second-order polynomial approximation with $10$ quantum gates gives an approximate to the actual potential energy function. Two proposed methodologies portray that one can do a trade-off between complexity and accuracy while encoding the potential energy function on a quantum simulator. It is important to note that, although Hadamard basis encoding shows a very close approximation to the classical encoded potential energy, it may suffer with poor fidelity as compared to the polynomial approximated method while running the experiment on a real quantum computer.

\subsubsection{Reconstruction with $3^{rd}$ order polynomial}

\begin{figure}[htb!]
\centering
\subfloat[Potential energy encoding for $3^{rd}$ order polynomial approximation with $4$ qubit quantum circuit]{
\includegraphics[width=0.9\linewidth]{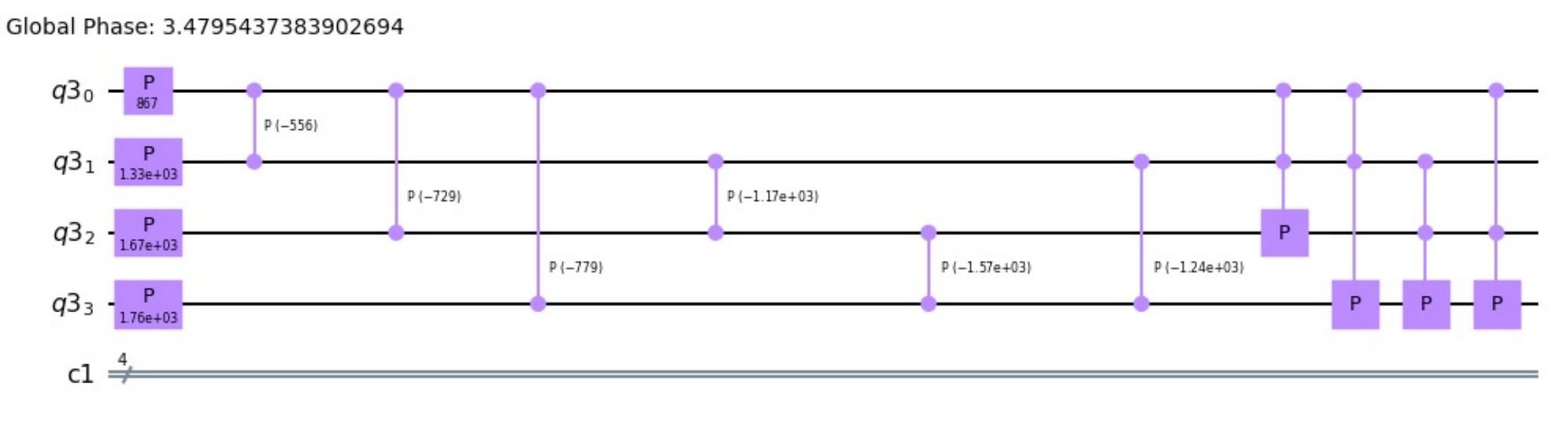} 
}
\hspace{0mm}
\subfloat[Potential energy encoding with $3^{rd}$ order polynomial approximation in $4$ qubit circuit: Imaginary part of $U_{PE}$ in vertical axis, discretized space ($x$) in horizontal axis]{
\includegraphics[width=0.49\linewidth]{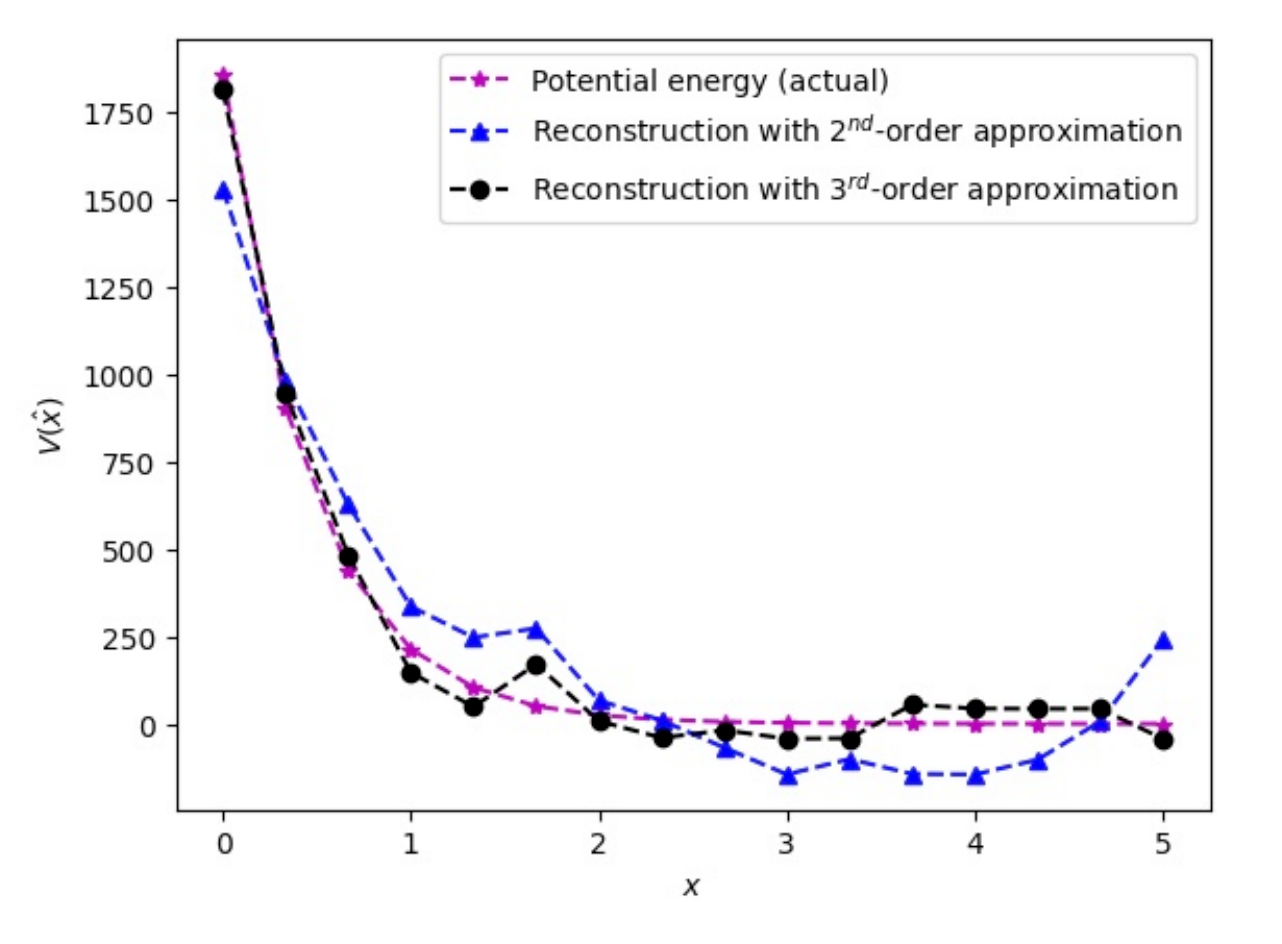} 
}
\caption{Potential energy of NaI reconstruction using $3^{rd}$ polynomial approximation}
\label{NaI_reconstruction}
\end{figure}

The potential energy reconstruction with a $4$ qubit quantum circuit embedding quantum gates for the $3^{rd}$ approximation is shown in Fig. \ref{NaI_reconstruction}. The quantum circuit with $4$ input qubit register as shown in Fig.
\ref{NaI_reconstruction}.$(a)$ possesses $4$ phase and $6$ controlled phase gates (same as the $2^{nd}$ order approximation). In addition, here we have added another $4$ CCp gates (which are $3$-input quantum gates) parameterised from the least square solution. The performance in the presence of the CCp gate compared with that of $2^{nd}$ order approximation is shown in Fig.\ref{NAI_potential}.$(b)$. It is observed that adding higher-order gates has better reconstruction than that of lower-order approximation of the potential energy at the cost of additional resources. Note that, the above reconstruction is based on the Qasm simulator result (which is not the performance of the actual simulator). The hardware performance for potential energy reconstruction is given in Section-\ref{fidelity_Section}.



\subsection{Gate complexity}
To show a quantitative picture of the computational resources required by the two different encoding schemes, viz. Hadamard basis encoding, and polynomial approximation encoding, we have shown a plot in Fig. \ref{gate_c}. The gate complexity with Hadamard basis encoding is approximate $\Theta(2^n)$, and the $2^{nd}$ order polynomial basis encoding
is given by $\Theta(n^2)$. Varying the qubit size $n=2$ to $n=10$, the number of quantum gates used by these two methods, we have shown the plot in Fig. \ref{gate_c}.

\begin{figure}[htb!]
\centering 
\includegraphics[width=0.49\linewidth]{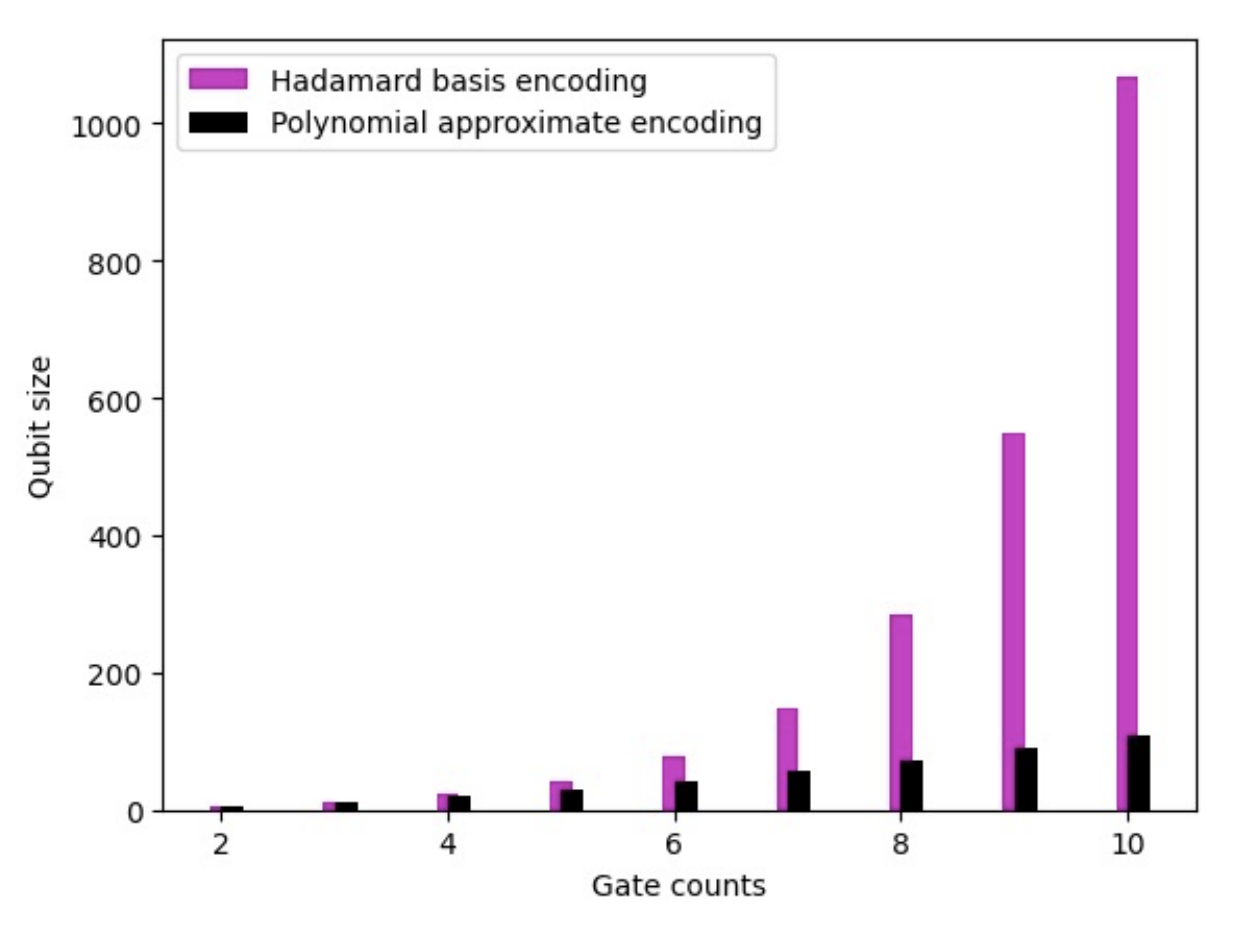} 
\caption{Gate complexity of the proposed encoding algorithms for potential energy encoding on a quantum circuit}
\label{gate_c}
\end{figure}

It is observed that the computational resources required by the polynomial approximate encoding are much lesser than the Hadamard encoding for a higher number of qubits. However, the Hadamard basis encoding method has higher accuracy of construction (of the unitary time evolution operator) and reconstruction (of the potential energy function) as it uses $2^n$ basis to encode $2^n$ functional points. The polynomial approximate method is a trade-off between accuracy and complexity. For practical quantum encoding of the potential energy function on a real quantum machine, the polynomial approximate encoding can be beneficial as it uses fewer computational resources as compared to the Hadamard basis encoding, thereby improving the quantum fidelity and cost of the resources.

\subsection{Fidelity performance}\label{fidelity_Section}
The evolution due to a complex diagonal matrix preserves the squared amplitude or the probability distribution of the wave function, since $\vert ae^{i\theta} \vert ^2 = \vert a \vert ^2$. As such, the fidelity of the evolved state with respect to the initial state should ideally be $1.0$. However, because of decoherence, the fidelity goes down with circuit depth. In order to measure the effectiveness of implementing these algorithms in present day quantum hardwares, and the impact the circuit depth may have on the fidelity, we used a swap test to measure the fidelity.

\begin{table}[htb!]
    \centering
    \begin{tabular}{|c|c|c|}
    \hline
    \textbf{Hamiltonian Encoding} & \textbf{Qasm Simulator} & \textbf{Hardware (ibm\_nairobi)}\\
    \hline
    Hadamard Encoding     & $1.0$ & $0.28$\\
         \hline
    \textcolor{teal}{$2^{nd}$ order poly encoding}  &  $1.0$ & $\textcolor{teal}{0.46}$\\
    \hline
    \textcolor{blue}{$3^{rd}$ order poly encoding}  &  $1.0$  & $0.40$\\
    \hline
    \end{tabular}
    \caption{Fidelity performance for $3$ qubit system on Qasm simulator and ibm\_nairobi \cite{ibm-quantum}}
    \label{tab1}
\end{table}

\begin{table}[htb!]
    \centering
    \begin{tabular}{|c|c|c|}
    \hline
    \textbf{Hamiltonian Encoding} & \textbf{Qasm Simulator} &  \textbf{Hardware (ibmq\_mumbai)} \\
    \hline
    Hadamard Encoding     & $1.0$ &$0.18$\\
         \hline
      \textcolor{teal}{$2^{nd}$ order poly encoding}  &  $1.0$ & $\textcolor{teal}{0.33}$\\
    \hline
    \textcolor{blue}{$3^{rd}$ order poly encoding ($1$ CCp)}  & $1.0$  & $0.27$\\
    \hline
     \textcolor{blue}{$3^{rd}$ order poly encoding ($2$ CCp)}  & $1.0$  & $0.2$ \\
    \hline
     \textcolor{blue}{$3^{rd}$ order poly encoding ($3$ CCp)}  & $1.0$  & $0.12$\\
    \hline
     \textcolor{blue}{$3^{rd}$ order poly Encoding ($4$ CCp)}  &  $1.0$ & $0.09$\\
    \hline
    \end{tabular}
    \caption{Fidelity performance for $4$ qubit system on Qasm simulator and ibm\_mumbai \cite{ibm-quantum}}
    \label{tab2}
\end{table}

We initialised the reference circuit with all qubits at state at $\ket{0}$ for measuring the fidelity. Table-\ref{tab1} gives the fidelity performance for a $3$ qubit circuit. In a Qasm simulator, as expected, every method performs equally well and we get a perfect fidelity. For the $3$ qubit circuit, we ran our experiments on an IBM quantum machine, called ibm\_nairobi. It is observed that the proposed $2^{nd}$ order approximate polynomial encoding method outperforms Hadamard encoding and higher-order approximation with a fidelity of $0.46$. 

We did a similar experiment for a $4$ qubit circuit on another IBM machine, called ibmq\_mumbai. $4$ qubit circuit, with larger depth as compared to the $3$ qubit circuit, sees a larger drop in fidelity. Table \ref{tab2} gives the details. The Hadamard encoding method achieves a fidelity of $0.18$, whereas the $2^{nd}$ order polynomial encoding has an improved fidelity of $0.33$. We have also shown the fidelity of $3^{rd}$ order polynomial approximation encoding.

The above experiments on real quantum hardwares show a comparative study of the proposed quantum encoding methods for encoding the potential energy function. Clearly, there are trade-offs one needs to consider when going for more accurate algorithms in a hardware.  


\section{Conclusion}
In this article, we have demonstrated a quantum Hamiltonian encoding framework for potential energy functions. The quantum time evolution operator consists of the potential energy and the kinetic energy parts. Here, we have discussed the potential energy function part. This will be an important ingredient in studying the evolution of wave functions of molecules, given the potential fields it experiences. This helps us in studying ionic and covalent chemical bonds, their formation and dissociation, etc. 

We have proposed two different encoding schemes here. The first method comprises combinations of Pauli-$Z$ and identity operators (that generate the Hadamard basis) and this helps us in encoding the potential energy functions with a very high degree of accuracy but with a complexity of $2^n$ for a $n$-qubit quantum circuit. In the  other proposed method, we have shown a polynomial approximation-based least square Hamiltonian encoding technique to reduce the computational gate complexity for practical usage in present-day noisy machines. We have shown the mathematical formulation for both schemes and discussed experimental results obtained on IBM quantum simulators and real hardware.

\section{Acknowledgement}

We are thankful to IBM Quantum for the overall support in this work, especially in giving access to the quantum machines for the research results. We acknowledge Dr. Anupama Ray, Dhiraj Madan, and Dr. SheshaShayee K Raghunathan of IBM Research, Bangalore, and Prof. Amit Kumar Dutta from IIT Kharagpur for their valuable suggestions. We would also like to thank Rajiv Sangle of IISc, Bangalore for the help in running the fidelity experiments.  



\bibliographystyle{IEEEtran}
\bibliography{Bibliography}%

\end{document}